%% file: dofx.tex
\documentclass[onecolumn, 11 pt, journal]{IEEEtran}
\usepackage{times}

\usepackage{amssymb}
\usepackage{amsmath}

\usepackage{epic,eepic}
\usepackage{epsfig}
\usepackage{cite}
\usepackage{verbatim}
\usepackage{enumerate}
\usepackage{subfigure}
\usepackage{multicol}
\usepackage{picinpar}
\usepackage{pstricks}
\usepackage{pst-node}

\def\QED{\mbox{\rule[0pt]{1.5ex}{1.5ex}}}

\newcommand{\define}{\stackrel{\triangle}{=}}

\newtheorem{corollary*}{\bf Corollary}
\newtheorem{theorem}{\bf Theorem}

\newtheorem{corollary}{\bf Corollary}
\newtheorem{lemma}{\bf Lemma}

\newcommand{\xH}{\mathbf{H}}
\newcommand{\xX}{\mathbf{X}}
\newcommand{\xY}{\mathbf{Y}}
\newcommand{\xZ}{\mathbf{Z}}
\newcommand{\xT}{\mathbf{T}}
\newcommand{\xv}{\mathbf{v}}
\newcommand{\xu}{\mathbf{u}}
\newcommand{\xw}{\mathbf{w}}
\newcommand{\xV}{\mathbf{V}}

\newcommand{\barxv}{\mathbf{\bar{v}}}
\newcommand{\barxH}{\mathbf{\bar{H}}}
\newcommand{\barxX}{\mathbf{\bar{X}}}
\newcommand{\barxY}{\mathbf{\bar{Y}}}
\newcommand{\barxZ}{\mathbf{\bar{Z}}}
\newcommand{\barxu}{\mathbf{\bar{u}}}

\begin{document}
\setcounter{page}{1}
\title{Degrees of Freedom of Wireless $X$ Networks}
\author{\authorblockN{Viveck R. Cadambe, Syed A. Jafar\\}
\authorblockA{Center for Pervasive Communications and Computing\\Electrical Engineering and Computer Science\\
University of California Irvine, \\
Irvine, California, 92697, USA\\
Email: {vcadambe@uci.edu, syed@uci.edu}\\ \vspace{-1cm}}}

\maketitle
\vspace{10pt}
\begin{abstract} 
We explore the degrees of freedom of  $M\times N$ user wireless $X$ networks, i.e. networks of $M$  transmitters and $N$ receivers where every transmitter has an independent message for every receiver. We derive a general outerbound on the degrees of freedom \emph{region} of these networks. When all nodes have a single antenna and all channel coefficients vary in time or frequency, we show that the \emph{total} number of degrees of freedom of the $X$ network is equal to $\frac{MN}{M+N-1}$ per orthogonal time and frequency dimension. Achievability is proved by constructing interference alignment schemes for $X$ networks that can come arbitrarily close to the outerbound on degrees of freedom. For the case where either $M=2$ or $N=2$ we find that the degrees of freedom characterization also provides a capacity approximation that is accurate to within $O(1)$. For these cases the outerbound is exactly achievable. While $X$ networks have significant degrees of freedom benefits over interference networks when the number of users is small, our results show that as the number of users increases, this advantage disappears. Thus, for large $K$, the $K\times K$ user wireless $X$ network loses half the degrees of freedom relative to the  $K\times K$ MIMO outerbound achievable through full cooperation. Interestingly, when there are few transmitters sending to many receivers ($N\gg M$) or many transmitters sending to few receivers ($M\gg N$), $X$ networks are able to approach the $\min(M,N)$ degrees of freedom possible with full cooperation on the $M\times N$ MIMO channel. Similar to the interference channel, we also construct an example of a $2$ user $X$ channel with propagation delays  where  the outerbound on degrees of freedom is achieved through interference alignment based on a simple TDMA strategy. 
\end{abstract}
\newpage
\section{Introduction}

There is increasing interest in approximate and/or asymptotic capacity characterizations of wireless networks as a means to understanding their performance limits. Two complementary approaches have been particularly successful in producing parsimonious characterizations of capacity along meaningful asymptotes. The first approach, which we refer to as the Gupta-Kumar approach, identifies the scaling laws of the capacity of wireless networks as the number of nodes approaches infinity. Starting with Gupta and Kumar's seminal work \cite{Gupta_Kumar}, this line of research has reached a level of maturity where the network scaling laws are known from both a mathematical \cite{Ozgur_Leveque_Tse} and physical perspective \cite{Franceschetti_Migliore_Minero}. Interesting ideas to emerge out of the Gupta-Kumar approach include the impact of mobility \cite{Grossglauser_Tse,Diggavi_Grossglauser_Tse}, hierarchical cooperation \cite{Ozgur_Leveque_Tse} and physical propagation models \cite{Franceschetti_Migliore_Minero}. The second approach is the degrees-of-freedom approach which considers a network with a \emph{fixed} number of nodes and identifies the capacity scaling laws as the signal to noise ratio is increased. The capacity scaling laws obtained through this approach are equivalently described by various researchers as  the multiplexing gain, the pre-log term or the degrees of freedom characterization. Starting with the point to point MIMO channel \cite{Foschini_Gans, Telatar}, the degrees of freedom have been characterized for MIMO multiple access \cite{Tse_Viswanath_Zheng}, MIMO broadcast channel  \cite{Yu_Cioffi, DimacsViswanath,Vishwanath_Jindal_Goldsmith}, 2 user MIMO interference channel  \cite{Jafar_Fakhereddin},  various distributed relay networks \cite{Boelcskei_Nabar_Oyman_Paulraj,Morgenshtern_Boelcskei_Nabar,Borade_Zheng_Gallager}, 2 user MIMO $X$ channel \cite{MMK, arxiv_dofx2, MMKreport1, MMKreport2, Jafar_Shamai}, and most recently the $K$ user interference channel  \cite{Cadambe_Jafar_int}. The impact of channel uncertainty on the degrees of freedom  has been explored in \cite{Jafar_mobile,Lapidoth_Shamai_collapse,lapidoth_dof, lapidoth, Jindal}. The degrees of freedom perspective has been used to characterize  the capacity benefits of cooperation between transmitters and/or between receivers through noisy channels in \cite{MadsenIT, Nosratinia_Madsen}.  The degrees of freedom benefits of cognitive message sharing are explored in \cite{Jafar_Shamai,Cadambe_Jafar_int,Lapidoth_Shamai_Wigger_IN,Devroye_Sharif}. In this work we explore  the degrees of freedom of $X$ networks.

\subsection*{$X$ Network}
\label{sec:sysmod}
\begin{figure}[h]
\begin{center}\input{X2tx3rx.eepic}\end{center}
\caption{A $2\times3$ user $X$ network}
\label{figure:X2tx3rx}
\end{figure}
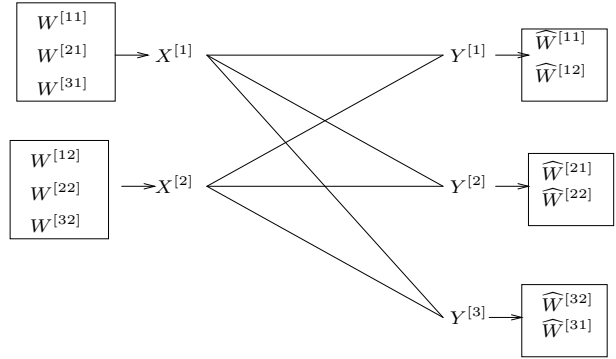
An $M\times N$ user $X$ network is a single-hop communication network with $M$ transmitters and $N$ receivers where each transmitter has an independent message for each receiver and all communication is one way - from the transmitters to the receivers. Thus, transmitters cannot receive and receivers cannot transmit which precludes relaying, feedback and cooperation between transmitters or cooperation between receivers. A $2\times 3$ user $X$ network is shown in Figure \ref{figure:X2tx3rx}. We refer to an $X$ network as a \textit{$K$ user $X$ network} if $M=N=K$. 
The $M \times N$ user $X$ network is described by input-output relations
$$Y^{[j]}(\kappa) = \displaystyle\sum_{i \in \{1,2 \ldots M\}} H^{[ji]}(\kappa) X^{[i]}(\kappa) + Z^{[j]}(\kappa),~~~ j=1,2 \cdots N$$
where $\kappa$ represents the channel use index. For simplicity we will assume $\kappa$ represents the time index. It should be noted that it can equivalently  be interpreted as the frequency index if coding occurs over frequency slots. $X^{[i]}(\kappa)$ is the signal transmitted by transmitter $i$ , $Y^{[j]}(\kappa)$ is the signal received by receiver $j$ and $Z^{[j]}(\kappa)$ represents the additive white Gaussian noise at receiver $j$. The noise variance at all receivers is assumed to be equal to unity. $H^{[ji]}(\kappa)$ represents the channel gain between transmitter $i$ and receiver $j$ at time $\kappa$. The channel coefficients are time varying and all nodes are assumed to have causal knowledge of all the channel gains. Also, to avoid degenerate channel conditions (identical/zero/infinite channel gains), we assume that all channel fade coefficients are drawn from a continuous distribution whose support lies between a non-zero minimum value and a finite maximum value. 

We assume that transmitter $i$ has message $W^{[ji]}$ for receiver $j$, for each $i \in \{1,2 \ldots M\}, j \in \{1,2 \ldots N\}$, resulting in a total of $M N$ mutually independent messages.  The total power across all transmitters is assumed to be $\rho$ per channel use. We denote the size of the message set by $|W^{[ji]}(\rho)|$. Let $ R_{ji}(\rho) = \frac{| W^{[ji]}(\rho)|}{\kappa_0} $ denote the rate of the codeword encoding the message $W^{[ji]}$, where the codeword spans $\kappa_0$ slots. A rate-matrix$[(R_{ji}(\rho))]$ is said to be \emph{achievable} if messages $W^{[ji]}$ can be encoded at rates $R_{ji}(\rho)$ so that the probability of error can be made arbitrarily small simultaneously for all messages by choosing appropriately long $\kappa_0$.
Let $C(\rho)$ represent capacity region of the $X$ network i.e it represents the set of all achievable rate-matrices $[(R_{ji}(\rho))]$. The degrees of freedom region of the $M\times N$ user $X$ network is defined by
\begin{eqnarray*}
\mathcal{D} = \Bigg\{ [(d_{ji})] \in \mathbb{R}_{+}^{MN}: \forall [(w_{ji})] \in \mathbb{R}_+^{MN} & &\\
		\displaystyle\sum_{i \in \{1,2 \ldots ,M\}, j \in \{1,2 \ldots ,N\}} w_{ji} d_{ji} & \leq & \lim\sup_{\rho \to \infty} \left[ \sup_{[(R_{ji}(\rho)) \in C(\rho)} \displaystyle\sum_{i \in \{1,2 \ldots ,M\}, j \in \{1,2 \ldots ,N\}} [ w_{ji} R_{ji}(\rho) ]\frac{1}{\log(\rho)} \right] \Bigg\} 
\end{eqnarray*}
$X$ networks are interesting because they encompass all communication scenarios possible in a one-way single hop wireless network. For example, multiple access, broadcast and interference networks are special cases of $X$ networks. Since there are messages from every transmitter to every receiver, every transmitter is associated with a broadcast channel, every receiver is associated with a multiple access channel and every disjoint pairing of transmitters and receivers comprises an interference channel within the $X$ network. In particular, any outerbound on the degrees of freedom of an $X$ network is also an outerbound on the degrees of freedom of all its subnetworks. This leads us to our first objective.

\noindent{\bf  Objective 1:} \emph{Find an outerbound on the degrees of freedom region of an $X$ network with $M$ transmitters and $N$ receivers.}

The total number of degrees of freedom $d$ \cite{Zheng_Tse} of a network is defined as the ratio of the sum capacity $C(\rho)$ to the logarithm of the SNR (signal-to-noise ratio) $\rho$ as the SNR grows large. 
\begin{eqnarray}
d=\lim_{\rho\rightarrow\infty}\frac{C(\rho)}{\log(\rho)}.
\end{eqnarray}
The degrees of freedom of a network provide a capacity approximation
\begin{eqnarray}
C(\rho) = d\log(\rho) + o(\log(\rho))
\end{eqnarray}
where $f(x)=o(g(x))$ means that $\lim_{x\rightarrow\infty}\frac{f(x)}{g(x)}=0$. The accuracy of this approximation approaches $100\%$ as the SNR approaches infinity. Note that the high SNR regime does not imply a high signal-to-interference-and-noise (SINR) ratio because both the signal and the interference powers are scaled. In fact, one can view the high SNR regime as simply de-emphasizing the local thermal noise relative to the signal (and interference) power. By de-emphasizing thermal noise, the degrees of freedom perspective directly addresses the issue of interference. One of the most interesting ideas to emerge out of the degrees of freedom perspective is the principle of \emph{interference alignment}. 

Interference alignment is best understood through a simple example presented in \cite{Cadambe_Jafar_int}. It is shown in \cite{Cadambe_Jafar_int} that in a wireless interference network where all transmissions interfere with each other, \emph{it is possible for everyone to communicate half the time with \emph{no interference} to each other}. If no interference is allowed, it is clear that the signals must be orthogonalized. The question, therefore, is whether it is possible to have orthogonal time division such that everyone communicates half the time. The problem may seem analogous to the task of dividing a cake among $K>2$ people so that everyone gets half of it - and hence quite impossible. However, the following example from \cite{Cadambe_Jafar_int} illustrates a simple scenario where everyone can commnicate half the time with no interference: Consider a wireless interference channel with $K$ transmitters, $K$ receivers and an independent message from each transmitter to its corresponding receiver. Suppose in this channel all desired links have a propagation delay equal to an even integer multiple of the basic symbol duration and all undesired cross links that carry interference have a propagation delay equal to an odd multiple of the basic symbol duration. In this channel, suppose all transmitters send simultaneously over the even time slots and are quiet during odd time slots. Because of the choice of propagation delays, the interference at each receiver aligns itself entirely over odd time slots so that the desired signals are received interference-free over even time slots. It is important to understand why the analogy with the cake cutting does not work. The reason is that in the interference channel, because of the propagation delays the time division seen by each receiver is different. The same signals that are orthogonal at one receiver may be perfectly aligned at another receiver. In the example above, the desired signal for each receiver is orthogonal to the interference as seen by that receiver. 

The propagation delay example described above is important for several reasons. First of all, it serves as a toy example to illustrate the key concept of interference alignment and the fallacy of the "cake-cutting" interpretation of spectral allocation. Second, it shows the significance of signal propagation delay on the capacity of wireless networks. Third, a direct relationship has been established between the delay propagation model and the deterministic channel model for wireless networks proposed in \cite{Avestimehr_Diggavi_Tse}. The deterministic channel model is a promising approach to understand the capacity of wireless networks and has been shown to lead to capacity approximations within a constant number of bits of the actual channel capacity in many interesting cases. Interference alignment on the deterministic channel is an intriguing problem. The propagation delay based interference alignment examples can be translated directly into the deterministic channel model. Then from the deterministic channel model these examples can be translated into interference alignment examples for real Gaussian networks with constant channel coefficients. Finally, the propagation delay example is also interesting also from a practical perspective. While the propagation delay  based interference alignment example is straightforward on the interference channel, no such example is known even for the $2$ user $X$ channel. As our second objective we address this problem.

\noindent{\bf Objective 2:} \emph{Construct an example of an $X$ network with propagation delays where the upperbound on the degrees of freedom is achieved through interference alignment based on time division.}

The idea of interference alignment emerged out of the work on the degrees of freedom of the $2$ user MIMO $X$ channel in \cite{MMK, arxiv_dofx2, MMKreport1, MMKreport2, Jafar_Shamai} and the compound broadcast channel in \cite{Weingarten_Shamai_Kramer}. A natural coding scheme for the $2\times 2$ user MIMO $X$ channel was proposed by Maddah-Ali, Motahari and Khandani in \cite{MMK, MMKreport1}. Viewing the $X$ channel as comprised of broadcast and multiple access components, the MMK scheme combines broadcast and multiple access coding schemes, namely dirty paper coding and successive decoding into an elegant coding scheme for the $X$ channel. A surprising result of \cite{MMK} is that $4$ degrees of freedom are achieved on the $2$ user $X$ channel with only $3$ antennas at each node. The key to this result is the implicit interference alignment that is achieved through iterative optimization of transmit precoding and receive combining matrices in the MMK scheme \cite{MMK}.  This observation lead to the first explicit interference alignment scheme, proposed in  \cite{arxiv_dofx2} which also shows that dirty paper coding or successive decoding (which may be advantageous at low SNR) are not required for degrees of freedom and that zero forcing suffices to achieve the maximum degrees of freedom possible on the $2\times 2$ user $X$ network. The zero forcing based interference alignment scheme of \cite{arxiv_dofx2} is used in \cite{MMKreport2} to characterize an achievable degree of freedom region for the $2\times 2$ user MIMO $X$ channel. The achievable degrees of freedom region is enlarged to include non-integer values in \cite{Jafar_Shamai}. \cite{Jafar_Shamai} also provides a converse to establish the optimality of the explicit interference alignment scheme in many cases.

Since its inception as a specific scheme for the MIMO $X$ channel, interference alignment has quickly found applications in a variety of scenarios, such as the compound broadcast channel \cite{Weingarten_Shamai_Kramer}, cognitive radio networks \cite{Jafar_Shamai, Devroye_Sharif}, deterministic channel models \cite{Jafar_Shamai_Cadambe} and most recently the interference channel with $K>2$ users \cite{Cadambe_Jafar_int}. The interference alignment scheme proposed in \cite{arxiv_dofx2} for the $2$ user MIMO $X$ channel  uses a simple spatial alignment of beamforming vectors. The compound broadcast channel \cite{Weingarten_Shamai_Kramer} and the MIMO $X$ channel with $M>1$ antennas at each node (where $M$ is not a multiple of $3$) \cite{Jafar_Shamai} introduce the time axis into the picture. The interference alignment scheme in these cases uses multiple symbol extension of the channel so that interference alignment is achieved through joint beamforming in both space and time. These  are also the first examples of wireless networks where non-integer degrees of freedom are shown to be optimal.  The interference alignment scheme in \cite{Jafar_Shamai} for the $2$ user $X$ with a single antenna at each node  also relies on channel variations in either time or frequency. The use of time/frequency variations is intriguing as it is not known whether these variations are essential to achieve the full degrees of freedom. The interference alignment scheme proposed in \cite{Cadambe_Jafar_int} for the interference channel with $K>2$ users and single antennas at all nodes requires a further generalization. In order to achieve the outerbound on the degrees of freedom, every user must achieve half his capacity at high SNR. With random channel coefficients the problem is over-constrained and even time variations and multiple symbol extensions are not found to be sufficient to achieve perfect interference alignment. Perfect interference alignment requires that each receiver should set aside exactly half of the observed signal space as a "waste basket". From each receiver's point of view, interfering transmissions are allowed as long as they are restricted to that recevier's designated interference (waste basket) signal space. Since each transmitter must use half the signal space dimensions in order to achieve half his capacity in the absence of interference at high SNR, the alignment problem requires a "tight" fit of all interfering signal spaces into the waste basket at each receiver, when each signal space has the same size as the waste basket itself. Further, this tight fit must be accomplished at every receiver. As it turns out this problem is over-constrained. The key to the solution proposed in \cite{Cadambe_Jafar_int} is to allow a few extra dimensions so that the problem becomes less constrained. The extra dimensions act as an overflow space for interference terms that do not align perfectly. It is shown in \cite{Cadambe_Jafar_int} that as the total number of dimensions grows large the size of the overflow space becomes a negligible fraction. Thus, for any $\epsilon>0$ it is possible to partially align interference to the extent that the total number of degrees of freedom achieved is within $\epsilon$ of the outerbound. The tradeoff is that the smaller the value of $\epsilon$ the larger the total number of dimensions (time/frequency slots) needed to recover a fraction $1-\epsilon$ of the outerbound value per time/frequency slot. Thus, we have the generalization of the interference alignment scheme to "partial" alignment with time varying channels and multiple symbol extensions. Interference alignment schemes have also been constructed with propagation delays (as in the example discussed earlier) in time dimension, Doppler shifts in frequency dimension, channel phase shifts in complex (real and imaginary) signal space, and in the codeword space with algebraic lattice codes \cite{Jafar_Shamai_Cadambe, Bresler_Parekh_Tse}.  Given the wide variety of generalizations, interference alignment is now best described as a "principle"   rather than a specific scheme. Thus, interference alignment refers to the general principle that signals may be chosen so that they cast overlapping shadows at the receivers where they constitute interference while they continue to be distinguishable at the receivers where they are desired. The key to interference alignment is that each receiver has a different view of the signal space, which makes it possible to have the signals align at one receiver and remain distinguishable at another receiver.

While the principle of interference alignment is quite simple, the extent to which interference can be aligned in a network is difficult to determine. Ideally one would like all interfering signals to align at every receiver and all desired signals to be distinguishable. On the $K$ user interference channel the transmitted  signal space of each user must align with the interference space of all other users' at their receivers but not at his own receiver. It is in fact quite surprising that enough interference alignment is possible that one can approach arbitrarily close to the outerbound ($K/2$) on the degrees of freedom. As we introduce more messages into the network, the interference alignment problem becomes even more challenging. The most challenging case for interference alignment is therefore the $X$ network where every transmitter has an independent message for every receiver. In this paper we explore this extreme scenario to find out the limits of interference alignment. This brings us to our next objective.

\noindent{\bf Objective 3:} \emph{Construct interference alignment schemes for the $X$ network with $M$ transmitters and $N$ receivers where the channel coefficients are random and time/frequency varying and each node is equipped with only one antenna.}

There is an important distinction between perfect interference alignment schemes and partial interference alignment schemes. Perfect interference alignment schemes are able to exactly achieve the degrees of freedom outerbound with a finite symbol extension of the channel. However, partial interference alignment schemes pay a penalty in the form of the overflow room required to "almost" align interference. While partial interference alignment schemes provide an innerbound that is within $o(\log(\rho))$ of the outerbound, only perfect interference alignment schemes are able to provide an innerbound that is within $O(1)$ of the capacity outerbound \cite{Cadambe_Jafar_int}. Since the degrees of freedom characterization  requires only an $o(\log(\rho))$ approximation to capacity both partial and perfect interference alignment schemes can provide tight innerbounds for degrees of freedom. However,  an $O(1)$ approximation is in general more accurate than an $o(\log(\rho))$ approximation. Therefore, it is interesting to identify the cases where the degrees of freedom characterization leads to an $O(1)$ characterization as well, as
\begin{eqnarray}
C(\rho) = d\log(\rho) + O(1).
\end{eqnarray}
$f(x)=O(g(x))$ means that $\lim_{x\rightarrow\infty}\frac{f(x)}{g(x)}<\infty$. As an example, consider the corresponding results for the $K>2$ user interference channel in \cite{Cadambe_Jafar_int}. For this channel, while the degrees of freedom are characterized in all scenarios where the channel coefficients are time varying, an $O(1)$ capacity characterization is found only for the $K = 3$ user case when all nodes have $M>1$ antennas. This is because a perfect interference alignment scheme was only found for the $K=3$ user case with $M>1$ antennas at all nodes. For the more general setting of the $X$ network we pursue a related objective.

\noindent{\bf Objective 4:} \emph{Identify $X$ network scenarios where degrees of freedom characterizations also provide O(1) capacity characterizations. In other words, identify $X$ network scenarios where perfect interference alignment is possible.}

The result that the $K$ user interference channel has $K/2$ degrees of freedom is interesting only for $K>2$ users. For $K=2$ users it is rather trivial as $K/2=1$ degree of freedom is achieved by simple TDMA between the two users. However, for the $X$ channel even the $2$ user case is interesting. Recall that the time varying $2$ user $X$ channel  with single antenna nodes has $4/3$ degrees of freedom \cite{MMK, Jafar_Shamai}. In other words, each of the $4$ messages on the $2$ user $X$ channel is able to access $1/3$ degrees of freedom. The observation that the $2$ user $X$ channel has a significant degrees of freedom advantage over the $2$ user interference channel leads us to the next objective:

\noindent{\bf Objective 5:} Determine the degrees of freedom advantage of the $K$ user $X$ channel over the $K$ user interference channel for $K>2$. 

A fundamental question for wireless networks is the capacity penalty of distributed signal processing. From the results of \cite{Cadambe_Jafar_int} we know that interference networks lose half the degrees of freedom compared to the MIMO outerbound corresponding to full cooperation. Thus, the loss of half the degrees of freedom is the cost of distributed processing in interference networks. Answering this question in the more general setting of $X$ networks is the last objective that we pursue in this paper.

\noindent{\bf Objective 6:} \emph{Characterize the cost, in terms of degrees of freedom, of distributed processing for $X$ networks with $M$ transmitters and $N$ receivers.}

Next we summarize the progress we make in this paper toward achieving these objectives.

\subsection{Overview of Results}

\subsubsection{Outerbound} The first result of this paper, presented in Section \ref{sec:outerbound},  is an outerbound for the degrees of freedom \emph{region} of the $M \times N$ user $X$ network. In particular, the \emph{total} number of degrees of freedom of the $M \times N$ user $X$ network is shown to be upper-bounded by $\frac{AMN}{M+N-1}$ per orthogonal time and frequency dimension, when each node is equipped with $A$ antennas. The outerbound is quite general as it applies to any fully connected (i.e. all channel coefficients are non-zero) $M\times N$ user $X$ network, regardless of whether the channel coefficients are constant or time varying. The key to the outerbound is to distribute the $MN$ messages in the $X$ network into $MN$ (partially overlapping) sets, each having $M+N-1$ elements. By picking these sets in a certain manner we are able to derive a MAC (multiple access channel) outerbound similar to \cite{Jafar_Fakhereddin} for the sum rate of the messages in each set. Since the MAC receiver has only $A$ antennas, the MAC has at most $A$ degrees of freedom. Thus, each set of messages can at most have $A$ degrees of freedom. The outerbounds for these sets together define an outerbound on the degrees of freedom region of the $M\times N$ user $X$ network and adding all the outerbounds gives us the bound on the total number of degrees of freedom. 
\subsubsection{Propagation Delay Example} We construct a TDMA based  interference alignment scheme that achieves the degrees of freedom outerbound for the $2\times 2$ user $X$ network with carefully chosen propagation delays for each link. As shown in \cite{Jafar_Shamai_Cadambe} for interference networks, this example can be easily applied to a deterministic channel model as well as to a Gaussian  $X$ channel model with single antennas at all nodes, no propagation delays and constant channel coefficients. This is the first known example of a Gaussian $X$ network with constant channel coefficients and single antenna nodes where the outerbound on degrees of freedom is achieved.
\subsubsection{Partial Interference Alignment Scheme} In Section \ref{sec:partialalign} we present a partial interference alignment scheme for $M\times N$ user $X$ networks with time varying channel coefficients. By considering larger supersymbols the partial interference alignment scheme is able to approach within any $\epsilon>0$ of the degrees of freedom outerbound. Combined with the outerbound, the partial interference alignment scheme establishes that the total number of degrees of freedom of $M\times N$ user $X$ networks with single antenna nodes and time (or frequency) varying channel coefficients is precisely $\frac{MN}{M+N-1}$. The partial interference alignment scheme does not extend completely to $X$ networks where each node has multiple antennas. However, if we imagine each antenna to be a separate user (which can only reduce the capacity) then a simple application of the partial interference alignment scheme shows that an innerbound of $\frac{AMN}{M+N-1/A}$ is achievable for $M\times N$ user $X$ networks where each node has $A$ antennas. If either $M$ or $N$ is reasonably large, then this innerbound is close to the outerbound.
\subsubsection{Perfect Interference Alignment Scheme} We construct a perfect interference alignment schemes for the $M\times N$ user $X$  channel when the number of receivers $N=2$. This scheme achieves exactly one degree of freedom for every message over an $M+N-1$ symbol extension of the channel, thus achieving exactly the outerbound of $\frac{MN}{M+N-1}$ total degrees of freedom over a finite channel extension. We also show an interesting \emph{reciprocity property} of beamforming and zero-forcing based schemes in wireless networks. In particular, we show that given a coding scheme in the $X$ network based entirely on beamforming and zero-forcing, we can construct a beamforming and zero-forcing based coding scheme over the reciprocal $X$ network achieving the same number of degrees of freedom as the original scheme. The coding scheme over the reciprocal channel may need \textit{apriori} knowledge of all channel gains even when the original scheme needs only causal channel knowledge. The reciprocal scheme is therefore practical in a scenario where channel extensions are considered in the frequency domain. This reciprocity property serves as an achievability proof for the $M \times N$ user $X$ channel when $M=2$, for any $N$. Thus, for either $M=2$ or $N=2$ we are able to construct perfect interference alignemnt schemes and in both these cases we have an $O(1)$ capacity characterization.
\subsubsection{$X$ networks versus Interference Networks} Since we are able to characterize the exact degrees of freedom of $X$ networks and the degrees of freedom of interference networks are already known, the comparison follows simply as a corollary. The $K\times K$ user $X$ channel has significant degrees of freedom advantage over the $K$ user interference channel when $K$ is small. However the advantage disappears as $K$ increases. This is easily seen by substituting $M=N=K$ in the total degrees of freedom expression for the $X$ channel to obtain $\frac{K^2}{2K-1}$ which is close to $K/2$ for large $K$. 
\subsubsection{Cost of Distributed Processing} This result also follows as a corollary of the main result that establishes the degrees of freedom for $X$ networks. Compared to $M\times N$ MIMO which represents joint signal processing at all transmitters and all receivers, the $M\times N$ user $X$ channel pays a degrees of freedom penalty of $\min(M,N) - \frac{MN}{M+N-1}$, which is the cost of distributed processing on the $X$ channel. While the cost of distributed processing is equal to half the degrees of freedom on the interference channel, it is interesting to note that for $X$ networks, this penalty disappears when the number of transmitters is much larger than the number of receivers or vice versa. This is easily seen because, when $M\gg N$ or $N\gg M$, then $\frac{MN}{M+N-1}$ is very close to $\min(M,N)$. In other words, a small set of distributed nodes in a wireless communication network with no shared messages can serve as a multi-antenna node, if they are transmitting to, or receiving from a large number of distributed nodes. We also provide an application of this result - the two-hop parallel relay network with $M$ distributed transmitting and receiving nodes with large number of relays. In \cite{Boelcskei_Nabar_Oyman_Paulraj}, this parallel relay network is shown to have $M/2$ degrees of freedom if the number of relays was large. By treating the network as a compound of a $M \times K$ and a $ K \times M$ $X$ channel, we construct an alternate degrees-of-freedom-optimal achievable scheme in section \ref{sec:parallelrelay}. 
   
\vspace{5pt}
\section{Degrees of Freedom Region Outerbound for $X$ Networks}
\label{sec:outerbound}
While our main focus in this paper is on the case where each nodes has a single antenna, we present the outerbound for the more general setting where transmitter $i$ has $A^t_i$ antennas and receiver $j$ has $A^r_j$ antennas, $\forall i\in\{1,2,\cdots, M\}, j\in\{1,2,\cdots,N\}$.
\begin{theorem}
\label{thm:outerbound}
Let 
\begin{eqnarray*} \mathcal{D}^{out} \define \bigg\{ [(d_{ji})] : \forall (m,n) \in \{1,2 \ldots M\} \times \{1,2 \ldots N\}\\
 \displaystyle\sum_{q=1}^{N} d_{qm} + \displaystyle\sum_{p=1}^{M} d_{np} - d_{nm} & \leq \max(A^t_m,A^r_n) & \bigg\} \end{eqnarray*}
Then $\mathcal{D} \subseteq \mathcal{D}^{out}$ where $\mathcal{D}$ represents the degrees of freedom region of the $M \times N$ $X$ channel
\end{theorem}
\begin{proof}
We start by defining $MN$ sets $\mathcal{W}^{nm},n\in\{1,2,\cdots,N\},m\in\{1,2,\cdots,M\}$ as follows:
\begin{eqnarray}
\mathcal{W}^{nm}\triangleq\{W^{[pq]}:(p-n)(q-m)=0\}
\end{eqnarray}
In other words, the set $\mathcal{W}^{nm}$ contains only those messages that either originate at transmitter $m$ or are destined for receiver $n$. Note that the $MN$ sets are not disjoint and that each set contains $M+N-1$ elements.

We will determine an outerbound for the total degrees of freedom achievable by each of the message sets when all other messages are eliminated. In other words, consider the $X$ channel when the only messages that need to be communicated are those that belong to the set $\mathcal{W}^{nm}$. Note that eliminating some messages cannot hurt the rates achievable by the remaining messages, as shown in \cite{Jafar_Shamai, Cadambe_Jafar_int}. Now we show that the total number of degrees of freedom of all messages in a set $\mathcal{W}^{nm}$ is no more than $\max(A^t_m,A^r_n) $.

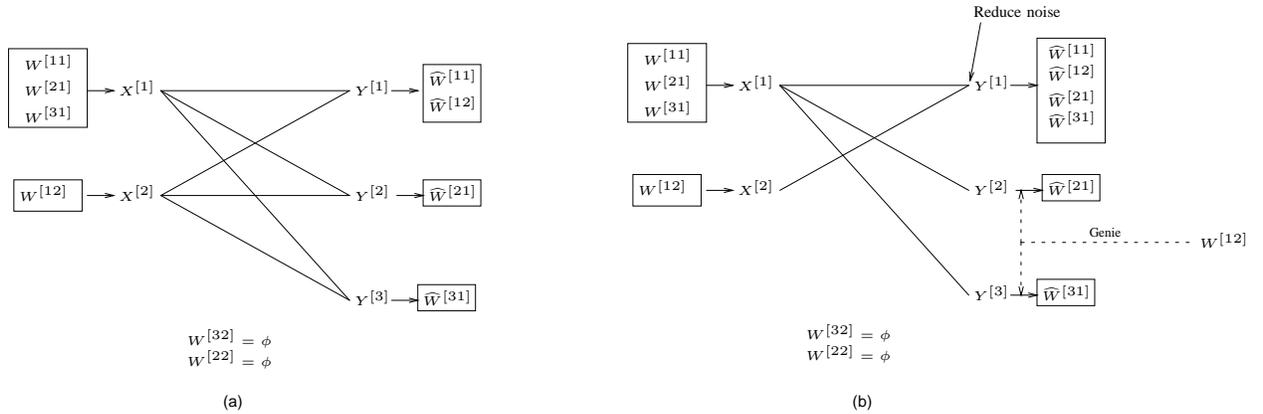
\begin{figure}[!bp]
\begin{center}\input{Xactivepairs.eepic}\end{center}
\caption{(a) $2 \times 3$ $X$ channel with pair $(1,1)$ active, (b) Converse argument in $2 \times 3$ $X$ channel with pair $(1,1)$ active}
\label{figure:activepair}
\end{figure}

Consider any reliable coding scheme in the $X$ channel where all messages not in the set $\mathcal{W}^{nm}$ are eliminated. Now, suppose a genie provides all the messages $W^{[nq]}, q\in\{1,2,\cdots,m-1,m+1,m+2,\cdots,M\}$ to each of the receivers $1,2,\cdots, n-1, n+1, n+2 \cdots N$. Then, receivers $1,2,3,\ldots,n-1,n+1, \ldots N$ can cancel the interference caused by $X^{[1]},X^{[2]},\cdots X^{[m-1]},X^{[m+1]} \ldots X^{[M]}$ so that, effectively, the receiver $p$ obtains $\hat{Y}^{[p]}$ from the received signal where
$$ \hat{Y}^{[p]} = H^{[pm]} X^{[m]} + Z^{[p]}$$
where $p\in\{1,2,\cdots,n-1,n+1,n+2,\cdots N\}$. 
Also using the coding scheme, receiver $n$ can decode its desired messages $W^{[nt]}, t=1,2,3 \ldots M$. Therefore, receiver $n$ can subtract the effect of $X^{[1]}, X^{[2]},\ldots,X^{[m-1]},X^{[m+1]}, \ldots X^{[M]}$ from the received signal so that it obtains $\hat{Y}^{[n]}$ where
$$ \hat{Y}^{[n]} = H^{[nm]} X^{[m]} + Z^{[n]}$$
Notice that receivers $p\neq n$ are able to decode messages $W^{[pm]}$ from $\hat{Y}^{[p]}$. Now, we can reduce the noise at receiver $n$ and if $A^r_n<A^t_m$ we add antennas at receiver $n$ so that it has $\max(A^t_m, A^r_n)$ antennas. By reducing noise and adding antennas we can ensure that $\hat{Y}^{[p]}, p\neq n$ are degraded versions of $\hat{Y}^{[n]}$ (for the details of this argument in the multiple antenna case, see \cite{Jafar_Fakhereddin}). In other words, by reducing noise and possibly adding antennas, we can ensure that receiver $n$ can decode all messages $W^{[pm]}$.  Note that the performance of the original coding scheme cannot deteriorate because of the genie or from reducing the noise or from adding antennas and therefore the converse argument is not affected. 
We have now shown that in a genie-aided channel with reduced noise (see Figure \ref{figure:activepair}), receiver $m$ is able to decode all the messages in the set $\mathcal{W}^{nm}$ when these are the only messages present. This implies that degrees of freedom of the messages in the set $\mathcal{W}^{nm}$ lies within the degrees of freedom region of the multiple access channel with transmitters $1,2 \ldots M$ and receiver $n$. Since receiver $n$ has $\max(A^t_m, A^r_n)$ antennas the total number of degrees of freedom for all messages in the set $\mathcal{W}^{nm}$ cannot be more than $\max(A^t_m, A^r_n)$. This gives us the outerbound
\begin{eqnarray}
\max_{d_{ij} \in \mathcal{D}}\sum_{q=1}^{N} d_{qm} + \displaystyle\sum_{p=1}^{M} d_{np} - d_{nm} & \leq \max(A^t_m, A^r_n)
\end{eqnarray}
Repeating the arguments for each $m,n$ we arrive at the result of Theorem \ref{thm:outerbound}.
\end{proof}

Since our focus in this paper is on the total degrees of freedom for the case when all nodes have one antenna, the following corollary establishes the needed outerbound.

\begin{corollary}
The total number of degrees of freedom of the $X$ channel with $M$ transmitters and $N$ receivers and $1$ antenna at each node, is upper bounded by $\frac{MN}{M+N-1}$ i.e.
\begin{eqnarray*} \max_{d_{ij} \in \mathcal{D}} \displaystyle\sum_{ij} d_{ij} \leq \frac{MN}{M+N-1} \end{eqnarray*}
\end{corollary}
\begin{proof}
The bound can be obtained by summing all the $MN$ inequalities describing the outerbound of the degrees of freedom region and setting $A^t_m=A^r_n=1$ for all transmitters and receivers.
\end{proof}

The outerbound of Theorem \ref{thm:outerbound} is not only useful for the total number of degrees of freedom, but rather it bounds the entire degree of freedom region of the $M\times N$ user $X$ network.  In other words, Theorem \ref{thm:outerbound} provides an outerbound for any fully connected distributed single hop network under the given system model. For example, consider a hypothetical channel with $3$ single antenna transmitters and $3$ single antenna receivers, and 6 messages $W^{[ij]}, i \neq j$. i.e the $3 \times 3$ user $X$ channel with $W^{[11]}=W^{[22]}=W^{[33]}=\phi$. The solution to the following linear programming problem provides an outerbound for the total number of degrees of freedom of this channel.
\begin{eqnarray*} 
 & \max_{d_{ij}} \displaystyle\sum_{i \neq j} d_{ij}  & \\
 \textrm{ s.t } & \displaystyle\sum_{q=1}^{3} d_{mq} + \displaystyle\sum_{p=1}^{3} d_{pl} - d_{ml}  \leq 1 & \forall (m,l) \in \{1,2,3\}\times \{1,2,3\}  
\end{eqnarray*}
In many cases of interest these outerbounds can be shown to be tight. For example, in the $2\times 2$ user $X$ network, the outerbound of Theorem \ref{thm:outerbound} is shown to represent the entire degrees of freedom region \cite{Jafar_Shamai}.

\section{Achievable Schemes - Propagation Delay Example}\label{sec:Xdelay}
Next we pursue Objective $2$ - to construct an $X$ network where a simple time division scheme will achieve the outerbound on the degrees of freedom by a careful choice of propagation delays. Note that, like the classical interference channel model, the $X$ network model we consider in the rest of this paper assumes zero propagation delays. However, we make an exception in this section by considering non-zero propagation delays to create an illustrative example. As explained in \cite{Jafar_Shamai_Cadambe}, this example is relevant to the case with zero delays as well.

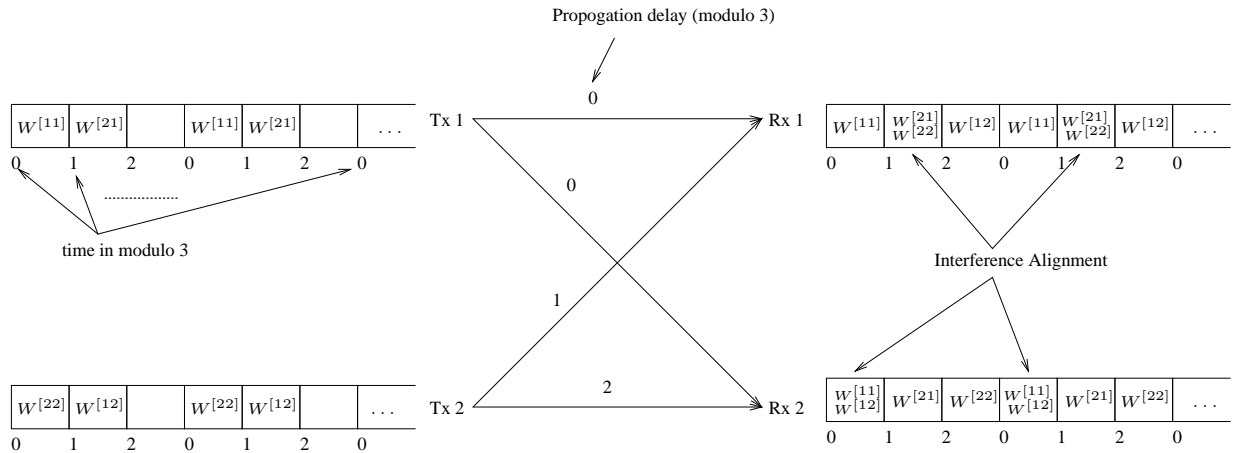
\begin{figure}[!tbp]
\begin{center}\input{2usrXdelay.eepic}\end{center}
\caption{Achieving 4/3 degrees of freedom over the $2$ user $X$ channel with propagation delays}
\label{fig:Xdelay}
\end{figure}

Consider a $2$ user $X$ channel with propagation delays between the transmitters and the receivers. Let $T_{ji}$ represent the propagation delay between transmitter $i$ and receiver $j$, where $i,j \in \{1,2\}$ and the delays are measured in units of the basic symbol duration. As usual, there are $4$ messages $W^{[ji]}$ in this $X$ channel, with $W^{[ji]}$ representing the message from transmitter $i$ to receiver $j$. Now, suppose the locations of the transmitters and receivers can be configured so that the propagation delays $T_{ji}$ satisfy the following relations. 
$$T_{11} = 3k, \mbox{ for some }k \in \mathbb{N} $$
$$T_{12} = 3l+1, \mbox{ for some }l \in \mathbb{N} $$
$$T_{21} = 3m, \mbox{ for some }m \in \mathbb{N} $$
$$T_{22} = 3p+2, \mbox{ for some }p \in \mathbb{N} $$
On this $X$ channel, all $4$ messages are encoded separately and will reach their desired receiver free of interference. This is accomplished as follows:\\
\emph{Transmitter 1:}
\begin{itemize}
\item Transmits one codeword symbol for message $W^{[11]}$ starting at $t = 3n, \forall n$.
\item Transmits a codeword symbol for message $W^{[21]}$ starting at $t = 3n+1 , \forall n$.
\end{itemize}
\emph{Transmitter 2:}
\begin{itemize}
\item Transmits a codeword symbol for message $W^{[12]}$ starting at $t = 3n+1, \forall n$.
\item Transmits a codeword symbol for message $W^{[22]}$ starting at $t = 3n, \forall n$.
\end{itemize}

With this transmission strategy, receiver $1$ receives codeword symbols corresponding to messages $W^{[11]}$ and $W^{[12]}$ in time slots satisfying $t=3n$ and $t=3n+2$ respectively (see figure \ref{fig:Xdelay}). The signals corresponding to undesired messages at this receiver, i.e. messages $W^{[21]}$ and $W^{[22]}$,  overlap in time slots $t = 3n+1$. In other words, the interference from $W^{[21]}$ and $W^{[11]}$ are aligned at receiver $1$. Similarly, from the point of view of receiver $2$, desired messages $W^{[21]}$ and $W^{[22]}$ arrive at at time slots  $t=3n+1 $ and $t=3n+2$, respectively, and the interfering messages $W^{[11]}$ and $W^{[12]}$ overlap at time slot $t=3n$. In this transmission scheme, each message is \emph{on} for $1/3$ of the time, and is received interference free at the receivers. Therefore, the scheme achieves $1/3$ degrees of freedom \emph{per message}, to achieve a \emph{total} of $4/3$ degrees of freedom in the $X$ network. From the outerbound we know this is the maximum degrees of freedom possible for the $2\times 2$ user $X$ network.

\section{Interference Alignment and Innerbounds on the Degrees of Freedom}
\label{sec:achievability}
The following is the main result of this section.
\begin{theorem}
\label{thm:achievability}
The $M \times N$ user $X$ network with single antenna nodes has $\frac{MN}{M+N-1}$ degrees of freedom.
\end{theorem}

The converse for the theorem is already proved in the corollary to Theorem \ref{thm:outerbound}. 

The achievable scheme for the $X$ channels are based on interference alignment and zero-forcing. For general $M\times N$ user $X$ channel we provide a partial interference alignment based innerbound that approaches the outerbound as we increase the size of the supersymbols (channel extensions). While the degrees of freedom achieved by this scheme can come within $\epsilon$ of the degrees of freedom outerbound for any $\epsilon>0$, the two are never exactly equal. This is sufficient for a degree of freedom characterization, but it does not provide an $O(1)$ capacity characterization. In some cases (when either $M=2$ or $N=2$) we are able to create \emph{perfect} interference alignment schemes so that the degrees of freedom outerbound is \emph{exactly} achieved with a finite channel extension. 

We start with the perfect interference alignment scheme for the $M\times 2$ user $X$ network.
\subsection{Perfect Interference Alignment for the $M \times 2$ user $X$ network}
\label{subsec:achKx2}
The outerbound for the $M\times 2$ user $X$ channel states that it cannot achieve more than a total of $\frac{2M}{M+1}$ degrees of freedom. In this section we construct an interference alignment scheme which achieves exactly $\frac{1}{M+1}$ degrees of freedom for each of the $2M$ messages, thus exactly achieving the outerbound.

Consider a $M+1$ symbol extension of the channel formed by combining $M+1$ time slots into a super-symbol. This channel can be expressed as 
$$\xY^{[1]}(\kappa) = \displaystyle\sum_{m \in \{1,2 \ldots M\}} \xH^{[1m]}(\kappa) \xX^{[m]}(\kappa) + \xZ^{[1]}(\kappa)$$
$$\xY^{[2]}(\kappa) = \displaystyle\sum_{m \in \{1,2 \ldots M\}} \xH^{[2m]}(\kappa) \xX^{[m]}(\kappa) + \xZ^{[2]}(\kappa)$$
where $\xX^{[m]}(\kappa)$ is a $(M+1) \times 1$ column vector representing the $M+1$ symbol extension of the transmitted symbol $X^{[m]}$, i.e 
$$\xX^{[m]}(\kappa) \define \left[ \begin{array}{c} X^{[m]}(\kappa(M+1)+1) \\X^{[m]}(\kappa(M+1)+2)\\ \vdots \\ X^{[m]}((\kappa+1)(M+1)) \end{array}\right]$$ Similarly $\xY^{[j]}$ and $\xZ^{[j]}$ represent $M+1$ symbol extensions of the $Y^{[j]}$ and $Z^{[j]}$ respectively. 
$\xH^{[jm]}$ is a diagonal $(M+1)\times(M+1)$ matrix representing the $M+1$ symbol extension of the channel, i.e., 
$$ \xH^{[jm]}(\kappa) \define \left[ \begin{array}{cccc}  H^{[jm]}(\kappa(M+1)+1) & 0 & \ldots & 0\\
	0 & H^{[jm]}(\kappa(M+1)+2) & \ldots & 0\\
	\vdots & \cdots & \ddots & \vdots\\ 
	0 & 0& \cdots  & H^{[jm]}((\kappa+1)(M+1)) \end{array}\right] $$

We now describe an achievable scheme that achieves one degree of freedom $d_{jm}=1, j=1,2, m=1,2\ldots M$ for each message over this $M+1$ symbol extension, thus achieving a total of $2M$ degrees of freedom over $M+1$ symbols.

The encoding strategy is as follows. Transmitter $m$ encodes messages $W^{[1m]}$ and $W^{[2m]}$ as two independent streams $x^{[1m]}$ and $x^{[2m]}$ and respectively transmits these two streams along directions $\mathbf{v}^{[1m]}$ and $\mathbf{v}^{[2m]}$ i.e. 
$$\xX^{[m]} = x^{[1m]} \mathbf{v}^{[1m]} + x^{[2m]} \mathbf{v}^{[2m]} $$
The received message at receiver $j$ is 
$$\xY^{[j]} = \displaystyle\sum_{m = 1}^{M}  \xH^{[jm]} x^{[1m]} \mathbf{v}^{[1m]} + \displaystyle\sum_{m =1}^{M} \xH^{[2m]} x^{[2m]} \mathbf{v}^{[2m]} +  \xZ^{[j]}$$
where $j=1,2$.

Receiver $1$ decodes its $M$ desired messages by zero-forcing the all interference vectors  $\mathbf{v}^{[2m]}, m=1,2\ldots M$. Now to recover $M$ interference-free dimensions for desired signals from the $M+1$ dimensions of the received symbol, the dimension of the interference has to be no more than $1$. Therefore, vectors $\mathbf{v}^{[2m]},i=2,\ldots M$ are picked so that their corresponding interference terms at receiver $1$ perfectly align with the interference from transmitter $1$ - i.e $\xH^{[1m]} \mathbf{v}^{[2m]}$ lies along $\xH^{[11]} \mathbf{v}^{[21]}$ for all $m=2,\ldots,M$.
\begin{equation} \label{eqn:intalign1} \xH^{[1m]} \mathbf{v}^{[2m]}= \xH^{[11]} \mathbf{v}^{[21]}, m=2\ldots M \end{equation}
This ensures that, from the point of view of receiver $1$, all the interference terms $\xH^{[1m]} \mathbf{v}^{[2m]}$ lie along a single vector  $\xH^{[11]} \mathbf{v}^{[21]}$. Similarly, at receiver $2$, the maximum dimension of the interference is $1$. This is ensured by picking $\mathbf{v}^{[1m]},m\neq1$ as 
\begin{equation} \label{eqn:intalign2}\xH^{[2m]} \mathbf{v}^{[1m]}= \xH^{[21]} \mathbf{v}^{[11]}, m=2\ldots M\end{equation}
 
Now that we have ensured all interference is restricted to only one dimension at each receiver, this dimension can be zero-forced to eliminate all interference, leaving $M$ interference free dimensions to recover the $M$ desired messagesfor each receiver. What is needed is that the desired signal vectors are linearly independent of the interference. Therefore we need to pick $\xv^{[11]}$ and $\xv^{[21]}$ so that the following matrices are of full rank.
$$ \left [ \xH^{[11]} \xv^{[11]}~~~\xH^{[12]}\xv^{[12]} ~~~ \ldots ~~~\xH^{[1M]} \xv^{[1M]} ~~~ \xH^{[11]} \xv^{[21]}\right] \textrm{ at receiver 1} $$ 
$$ \left [ \xH^{[21]} \xv^{[21]}~~~\xH^{[22]}\xv^{[22]} ~~~ \ldots ~~~\xH^{[2M]} \xv^{[2M]} ~~~ \xH^{[21]} \xv^{[11]}\right] \textrm{ at receiver 2} $$
Note that the first $M$ column vectors of the above matrices represent the signal components, and the last column represents the aligned interference. We now pick the columns of $\xv^{[11]}$ and $\xv^{[21]}$ randomly from independent continuous distributions i.e.,
$$\mathbf{v}^{[11]} = \left[ \begin{array}{c} p_1 \\ p_2 \\ \vdots \\ p_{M+1} \end{array} \right]$$
$$\mathbf{v}^{[21]} = \left[ \begin{array}{c}  q_1 \\ q_2 \\ \vdots \\ q_{M+1}\end{array} \right]$$
An important observation here is that once $\xv^{[11]}$ and $\xv^{[21]}$ are picked as above, equations (\ref{eqn:intalign1}) and (\ref{eqn:intalign2}) can be used to pick $\xv^{[1m]}, m \neq 1$ with just causal channel knowledge. i.e.,  the $l$th component of the transmitted signal at any transmitter depends only on the first $l$ diagonal entries of $\xH^{[ji]}$ (in fact, it depends only on the $l$th diagonal entry ). The desired signal can now be shown to be linearly independent of the interference at both receivers almost surely. 

For example at receiver $1$, we need to show that matrix 
$$ \mathbf{\Lambda} = \left [ \xH^{[11]} \xv^{[11]}~~~\xH^{[12]}\xv^{[12]} ~~~ \ldots ~~~\xH^{[1K]} \xv^{[1K]} ~~~ \xH^{[11]} \xv^{[21]}\right]$$
has full rank. Since all channel matrices are diagonal and full rank, we can multiply by $(\xH^{[11]})^{-1}$ and use equations \ref{eqn:intalign1} and \ref{eqn:intalign2} to replace the above matrix by 
$$ \mathbf{\Lambda} = \left [  \xv^{[11]}~~~~(\xH^{[11]})^{-1}\xH^{[12]}\xH^{[21]} (\xH^{[22]})^{-1} \xv^{[11]} ~~~~ \ldots ~~~~ (\xH^{[11]})^{-1}\xH^{[1K]} \xH^{[21]} (\xH^{[2K})^{-1}  \xv^{[11]} ~~~~ \xv^{[21]}\right]$$
Now, notice that an element in the $l$th row of the above matrix is a monomial term in the random variables $H^{[1i]}_l, H^{[2i]}_l, p_l, q_l, i=1,2 \ldots M$ where $H^{[ji]}_l$ represents the diagonal entry in the $l$th row of $\xH^{[ji]}$. Also, all the monomial terms are unique, since $H^{[1m]}_l$ has power $1$ in the $m$th column and power $0$ in all other columns for $m=1,2 \ldots M$. The terms in the $M+1$th column are also unique since it is the only column with a positive power in $q_l$. Therefore, Lemma \ref{lemma:nonsingular} implies that the matrix has a full rank of $M+1$ almost surely.

Similarly, the desired signal can be shown to be linearly independent of the interference at receiver $2$ almost surely. Therefore, $2M$ independent streams are achievable over the $(M+1)$ symbol extension of the channel implying $\frac{2M}{M+1}$ degrees of freedom over the original channel.
Also, since the achievable scheme essentially creates $2M$ point-to-point links over a $M+1$ symbol extension of the channel, it provides an $\mathcal{O}(1)$ capacity characterization of this network as 
$$ C(\rho) = \frac{2M}{M+1} \log(\rho) + \mathcal{O}(1)$$
where $C(\rho)$ is the sum-capacity of the network as a function of transmit power $\rho$.

\subsection{Achievability for  $2 \times M$ $X$ network - Reciprocity of beamforming and zero-forcing based schemes}
\label{subsec:ach2xK}
Consider an $M\times N$ user $X$ network. We refer to this as the primal channel. Consider any achievable scheme on this channel based on beamforming and zero-forcing respectively. Specifically, consider any achievable scheme whose encoding strategy is of the form:
\begin{itemize}
\item Encoding - Transmitter $i$ encodes a message to receiver $j$ along linearly independent streams and beamforms these streams along linearly independent vectors. For example, the $k$th stream to receiver $i$ is encoded as $x^{[ji]}_{k}$ and beamformed along direction $ \xv^{[ji]}_k$  as $x^{[ji]}_k \xv^{[ji]}_k$.
\item Decoding - Receiver $j$ decodes all the desired streams through zero-forcing. For example, to decode the $k$th stream from transmitter $i$ i.e $x^{[ji]}_k$, the receiver projects the received vector along vector $\xu^{[ji]}_k$ which zero-forces all undesired streams.
\end{itemize}
The reciprocal (or dual) channel is the the channel formed when the transmitters and receivers of the primal channel are interchanged and the channel gains remain the same. Therefore, the dual of an $M \times N$ $X$ channel is a $N \times M$ $X$ channel. The channel gain between transmitter $i$ and receiver $j$ in the primal channel is equal to the channel gain between transmitter $j$ and receiver $i$ in the dual channel. It can be shown that corresponding to every zero-forcing based achievable scheme in the primal network, there exists a zero-forcing based achievable scheme in the reciprocal network that achieves the same number of degrees of freedom as the primal network. In particular, the coding scheme that achieves this in the dual network may be described as follows.
\begin{itemize}
\item Encoding - In the dual network, transmitter $j$ encodes a message to receiver $i$ along linearly independent streams and beamforms these streams along directions that were used for zero-forcing in the primal network. For example, the $k$th stream to receiver $i$ is encoded as $\bar{x}^{[ij]}_{k}$ and beamformed along direction $ \xu^{[ij]}_k$  as $x^{[ij]}_k \xu^{[ij]}_k$, ( where $\xu^{[ij]}_k$ represents the zero-forcing vector used in the primal network by receiver $j$ to decode the $k$th stream from transmitter $i$ )
\item Decoding - Receiver $i$ decodes all the desired streams through zero-forcing along directions that were the beamforming directions in the primal network. For example, in the dual network, to decode the $k$th stream from transmitter $j$ i.e $\bar{x}^{[ij]}_k$, receiver $i$ projects the received vector along vector $\xv^{[ij]}_k$ - where $\xv^{[ij]}_k$ represents the beamforming vector used in the primal network by transmitter $i$ to transmit the $k$th stream to receiver $j$ 
\end{itemize}
It can be easily verified that the above scheme maps every independent stream in the primal $M \times N$ $X$ channel to an independent decodable stream in the dual $N \times M$ $X$ channel and thus achieves the same number of degrees of freedom in the dual network. This scheme therefore establishes a duality of beamforming and zero-forcing based schemes in the general $X$ network. Note that in order to construct the zero-forcing vectors on the primal network, each transmitter in the dual network needs apriori knowledge of all the beamforming and channel vectors, and therefore violates the causality constraint in time-varying channels (i.e., if $\kappa$ represents the time-index). More fundamentally, causality only affects a transmitter in a communication network, since a receiver is ``willing to wait'' arbitrarily long before decoding a codeword. Therefore, the zero forcing scheme at the receivers of the primal network are not constrained by causality. Notice that in our coding scheme for the dual network, the encoding strategy at the transmitters in this dual network relies on the knowledge of the decoding strategy in the corresponding primal network and therefore, switching transmitters and receivers in the primal network to form the dual network affects causality constraints.
Below, we provide a formal proof of reciprocity for the $M=2$ case. The proof serves as an achievable scheme in the $2 \times K$ $X$ network with frequency selective channels.

\begin{proof}
Consider the $M+1$ symbol extension of the $M \times 2$ user $X$ network. The reciprocal network of this extended $X$ channel is a $2 \times M$ $X$ channel which can be expressed as 
$$\barxY^{[m]} = \displaystyle\sum_{l \in \{1,2 \ldots N\}} \barxH^{[ml]} \barxX^{[l]} + \barxZ^{[m]}, m=1,2 \ldots M$$
where the over-bar notation indicates quantities in the reciprocal channel. The vectors $\barxX^{[i}, \barxY^{[i]}, \barxZ^{[i]}$ are $(M+1) \times 1$ vectors and $\barxH^{[ij]}$ is $(M+1) \times (M+1)$ matrix. Note that in the reciprocal channel, the channel gains are identical i.e
$$\barxH^{[ij]} = \xH^{[ji]}, i=1,2 \ldots M, j=1,2$$

In the achievable scheme for the $M \times 2$ $X$ channel described in Section \ref{subsec:achKx2}, transmitter $i$ encodes message $W^{[ji]}$ as $x^{[ji]}$ and beamforms it along direction $\xv^{[ji]}$ so that the received message at receiver $j$ is 
$$\xY^{[j]} = \displaystyle\sum_{i=1}^{K}  x^{[1i]}\xH^{[ji]}  \mathbf{v}^{[1i]} + \displaystyle\sum_{i=1}^{K} x^{[2i]}\xH^{[2i]}  \mathbf{v}^{[2i]} +  \xZ^{[j]}$$
Receiver $j$ decodes $x^{[ji]}$ using zero-forcing. Let $\xu^{[ji]}$ represent the zero-forcing vector used by receiver $j$ to decode $x^{[ji]}$. Therefore
\begin{equation}\label{eqn:primalzeroforce} (\xu^{[ji]})^T \xH^{[jm ]} \xv^{[lm]} \neq 0 \Leftrightarrow l=j,m=i \end{equation}
$$ \Rightarrow (\xu^{[ji]})^T \xY^{[j]} = x^{[ji]} (\xu^{[ji]})^T \xH^{[ji]} \xv^{[ji]} + (\xu^{[ji]})^T \xZ^{[j]}$$

We now construct beamforming directions $\barxv^{[ij]}$ and zero-forcing vectors $\barxu^{[ij]}$ in the dual channel so that receiver $i$ can decode message $\bar{W}^{[ij]}$ by zero-forcing interference from all other vectors. 
In the dual channel, let transmitter $j$ encode message $\bar{W}^{[ij]}$ to receiver $i$ as $\bar{x}^{[ij]}$ where $j\in \{1,2\}, i \in \{1,2,\ldots M\}$. The beamforming directions of the dual channel are chosen to be the zero-forcing vectors in the primal channel i.e. $\barxv^{[ij]} = \xu^{[ji]}, j \in \{1,2\}, i \in \{1,2,\ldots M\}$. The transmitted message is therefore
$$\barxX^{[j]} = \displaystyle\sum_{m=1}^{M} \bar{x}^{[mj]} \barxv^{[mj]} = \displaystyle\sum_{m=1}^{M} \bar{x}^{[mj]} \xu^{[jm]}, j=1,2  $$

The received vector at receiver $i$ is 
$$\barxY^{[i]} = \displaystyle\sum_{m=1}^{M} \left( \bar{x}^{[m1]}\barxH^{[i1]}  \barxv^{[m1]} + \bar{x}^{[m2]}\barxH^{[i2]}  \barxv^{[m2]}\right)  + \barxZ^{[i]}$$ 
$$\barxY^{[i]} = \displaystyle\sum_{m=1}^{M} \left( \bar{x}^{[m1]}\xH^{[1i]}  \xu^{[1m]} + \bar{x}^{[m2]}\xH^{[2i]}  \xu^{[2m]}\right)  + \barxZ^{[i]}$$ 
Now, at receiver $i$, stream $\bar{x}^{[ij]}$ is decoded by projecting the received vector along $\xv^{[ji]}$
$$ (\xv^{[ji]})^T\barxY^{[i]} = \displaystyle\sum_{m=1}^{M} (\xv^{[ji]})^T \left( \bar{x}^{[m1]}\xH^{[1i]}  \xu^{[1m]} + \bar{x}^{[m2]}\xH^{[2i]}  \xu^{[2m]}\right)  + (\xv^{[ji]})^T \barxZ^{[i]}$$ 
$$ (\xv^{[ji]})^T\barxY^{[i]} = \displaystyle\sum_{m=1}^{M}  \left( \bar{x}^{[m1]} (\xv^{[ji]})^T\xH^{[1i]}  \xu^{[1m]} + \bar{x}^{[m2]} (\xv^{[ji]})^T\xH^{[2i]}  \xu^{[2m]}\right)  + (\xv^{[ji]})^T \barxZ^{[i]}$$ 
$$ (\xv^{[ji]})^T\barxY^{[i]} = \displaystyle\sum_{m=1}^{M} \left( \bar{x}^{[m1]} (\xH^{[1i]} \xu^{[1m]})^T \xv^{[ji]}  + \bar{x}^{[m2]}  (\xH^{[2i]} \xu^{[2m]})^T \xv^{[ji]}\right)  + (\xv^{[ji]})^T \barxZ^{[i]}$$ 
$$ (\xv^{[ji]})^T\barxY^{[i]} = \displaystyle\sum_{m=1}^{M} \left( \bar{x}^{[m1]} (\xu^{[1m]})^T \xH^{[1i]} \xv^{[ji]} + \bar{x}^{[m2]}  (\xu^{[2m]})^T \xH^{[2i]} \xv^{[ji]} \right)  + (\xv^{[ji]})^T \barxZ^{[i]}, j=1,2, i=1,2, \ldots,M$$ 
Using equation (\ref{eqn:primalzeroforce}) above, we get
$$ (\xv^{[ji]})^T\barxY^{[i]}  = \bar{x}^{[ij]} (\xu^{[ji]})^T \xH^{[ji]} \xv^{[ji]} + (\xv^{[ji]})^T \barxZ^{[i]}$$
The zero-forcing vector $\xv^{[ji]}$ cancels all interference. Thus,  $1$ degree of freedom is achieved for message $W^{[ij]}$, for each $i \in \{1,2 \ldots M\}, j \in \{1,2\}$. Thus $2M$ degrees of freedom are achieved over the extended $2 \times M$ user $X$ channel implying that $\frac{2M}{M+1}$ degrees of freedom are achieved in the original $2 \times M$ $X$ channel.
The reciprocal scheme also implies that the $O(1)$ capacity of the $2 \times M$ user $X$ network is
$$ C(\rho) = \frac{2M}{M+1} \log(\rho) + \mathcal{O}(1)$$
where $C(\rho)$ is the sum-capacity of the network as a function of transmit power $\rho$
\end{proof}
\subsection{Partial Interference Alignment for General $M\times N$ user $X$ Networks}\label{sec:partialalign}
In this section we provide a general interference alignment scheme that applies to the $M\times N$ user $X$ network for any $M,N$. Unlike the special cases of $M=2$ and $N=2$ considered in the previous section where the degree of freedom outerbound is exactly achieved, the general scheme is a partial interference alignment scheme and can only come arbitrarily close the the degrees of freedom outerbound by using long channel extensions.

We start with an interpretation of the outerbound in terms of interference alignment - this interpretation provides an intuition of the proof for the general case. For the rest of the discussion in this section, we restrict our achievable scheme to the class of schemes that use beamforming and zero-forcing over extended channels. 

Consider an $X$ network with $M$ transmitters and $N$ receivers. If at all possible, consider a hypothetical coding scheme based on interference alignment and zero-forcing, that achieves the outerbound of Theorem \ref{thm:outerbound}. Note that $d_{ij} = \frac{1}{M+N-1}, \forall i\in\{1,2 \ldots M\},j \in \{1,2 \ldots N\}$ is the only point that satisfies all the outerbounds and achieves a total of $\frac{MN}{M+N-1}$ degrees of freedom. Therefore all messages achieve equal number of degrees of freedom in this coding scheme. Assuming that the achievable scheme uses $(M+N-1)k$ symbol extension of the channel, the scheme achieves $k$ degrees of freedom for each message over this extended channel. Now, transmitter $i$ encodes message $W^{[ji]}$ for message $j$ as $k$ independent streams. The transmitted symbol is a superposition of all these streams. 

Consider the signal received at receiver $1$. The received signal consists of $Nk$ linearly independent vectors from each transmitter. For example, consider the signal received from transmitter $1$. Among the $Nk$ linearly independent streams received from transmitter $1$, the desired signal forms $k$ streams so that the dimension of the interference from transmitter $1$ is equal to $(N-1)k$. Now, note that receiver $1$ has to decode $Mk$ streams, ( $k$ streams corresponding to each transmitter ) from the from the $(M+N-1)k$ components of the received symbol; therefore, the dimension of the interference has to be not more than $(N-1)k$. Since, the receiver already receives $(N-1)k$ linearly independent interference vectors from transmitter $1$, all the interference vectors from all the other transmitters have to align perfectly with the interference from transmitter $1$. This argument thus leads to the following insight : the outer-bound provided in Theorem \ref{thm:outerbound} corresponds to the scenario where, at any receiver, all the interference from transmitters $2,3, \ldots M$ perfectly aligns with the interference from transmitter $1$. The achievable scheme of Theorem \ref{thm:achievability} is based on precisely this insight; the construction of Lemma \ref{lemma:intalign} (stated below) is used to design transmit beamforming directions over arbitrarily long channel extensions, so that at each receiver, all interference from transmitters $2,3 \ldots M$ aligns  within the interference from transmitter $1$ to achieve partial interference alignment in the $X$ network.

In order to show achievability for general $M,N$, we use the following two lemmas.
\begin{lemma}
\label{lemma:nonsingular}
Consider an $M \times M$ square matrix $\mathbf{A}$ such that $\forall i\in\{1,2,\cdots,M\}$, every element $a_{ij}, j\in\{1,2,\cdots,M\}$ in the $i^{th}$ row of $\mathbf{A}$  is a \emph{different}\footnote{We refer to two monomial terms as \emph{different} if there exists atleast one variable with different exponents in the two monomial terms} multivariate monomial term in the variables $X_{ik},  k\in\{1,2,\cdots,K\}$. If $X_{ik}$ has a continuous probability distribution conditioned on all $X_{pq}, \forall (p,q) \neq (i,k)\}$ then the matrix has a full rank of $M$ with probability $1$.
\end{lemma}
The lemma is used to show the linear independence of the desired signals and the vector subspace that contains all the interference. The columns of the matrix $A$ represent desired signal vectors and the vectors that span the interference subspace. Therefore, if $A$ is full rank then its columns are linearly independent, which means the signals and the interference are separable. The received signal vectors as well as the interference vectors are functions of the channel coefficients such that each row of $A$ contains the channel coefficient values from a different time slot.
\begin{proof}
Refer to Appendix \ref{app:nonsingular}
\end{proof}

\begin{lemma}
\label{lemma:intalign}
Let $\xT_1,\xT_2, \ldots \xT_N$ be diagonal matrices of size $\mu \times \mu$, consisting of random variables on the diagonals such that
any diagonal element $\xT_i[j,j]$ has a continuous probability distribution, conditioned on all other elements $\xT_{l}[k,k], \forall (l,k)\neq (i,j)$. 
Also, let $\xw$ be a column vector whose elements are generated i.i.d. from a continuous distribution.
Then, for any $n \in \mathbb{N}$ satisfying $\mu > (n+1)^N$, we can construct, with probability $1$, full rank matrices $\xV$ and $\xV^{'}$ of sizes $\mu \times n^N$ and $\mu \times (n+1)^N$ respectively, such that the following relations are satisfied.
\begin{eqnarray*} 
\xT_1 \xV &\prec& \xV^{'}   \\
\xT_2 \xV &\prec& \xV^{'}   \\
	& \vdots & \\
\xT_N \xV &\prec& \xV^{'}   
\end{eqnarray*}
where $\mathbf{P} \prec \mathbf{Q}$ implies that the span of the column vectors of $\mathbf{P}$ lies in the vector space spanned by the column vectors of $\mathbf{Q}$. Furthermore, the above conditions can be satisfied with every entry in the $k$th row of $\xV$ ( and $\xV^{'}$ ) being a multi-variate monomial function of entries in the $k$th rows of $\xw$ and $\xT_i,i=1,2 \ldots N$. 

\end{lemma}
\begin{proof}
Refer to Appendix \ref{app:intalign}
\end{proof}

Note that Lemma \ref{lemma:intalign} provides a general approach to interference alignment. In particular, it generalizes the achievable scheme used in the $K$ user interference channel \cite{Cadambe_Jafar_int}. To see this, suppose vector spaces $\xV$ and $\xV^{'}$ transmitted by two distributed transmitters are perceived as interference at $N$ receivers. Then, the relations satisfied by $\xV$ and $\xV^{'}$ can be interpreted as the vector space $\xV$ aligning itself with vector space $\xV^{'}$ at these $N$ receivers. In fact, notice that for large values of $n$, $ \mbox{rank}(\xV) = n^N \approx (n+1)^N = \mbox{rank}(\xV^{'})$, meaning that the vector space $\xV$ aligns (asymptotically) perfectly with $\xV^{'}$ as $n \to \infty$. Also, a critical property used in the lemma is the commutative property of multiplication of diagonal matrices $\xT_j$. It turns out that the construction of the lemma does not extend directly to MIMO scenarios, since multiplication is not commutative if the matrices are not diagonal. The achievability proof for the total degrees of freedom of the general $M\times N$ user $X$ channel is placed in Appendix \ref{app:proof_dofx}.

\section{Degrees of Freedom of the Parallel Relay Network}
\label{sec:parallelrelay}
\begin{figure}[!tbp]
\begin{center}\input{parallel_relay.eepic}\end{center}
\caption{The parallel relay network with $M=2$}
\label{fig:parallelrelay}
\end{figure}
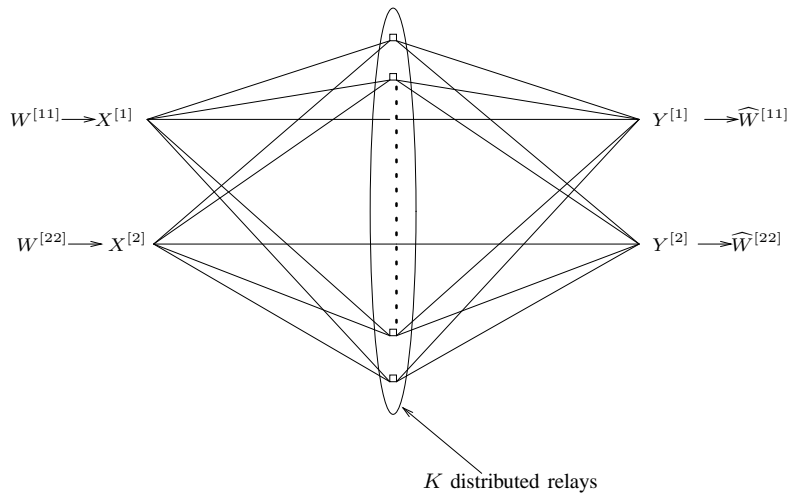

In this section, we present an application of the results of the previous section. We provide an alternate proof of the degrees of freedom characterization of the parallel relay network (Figure \ref{fig:parallelrelay}), first studied in \cite{Boelcskei_Nabar_Oyman_Paulraj}. 

Consider a two hop parallel relay network with $M$ distributed single antenna transmitters and $M$ distributed single antenna receivers. We assume that the intermediate hop has $K$ half-duplex relays. 


Much like the $M$ user interference channel, transmitter $j$ has a single message $W^{[j]}$ to transmit to receiver $j$, where $j=1,2,\ldots M$ and thus there are a total of $M$ messages in this channel. Through an achievable scheme based on amplify-and-forward strategy at the relays, \cite{Boelcskei_Nabar_Oyman_Paulraj} shows that this network has $M/2$ degrees of freedom if the number of relays $K$ grow arbitrarily large. We use the degrees of freedom characterization of $X$ channel to provide an alternate optimal achievable scheme to show the same result (Figure \ref{fig:parallelrelay}) by treating this network as a combination of two $X$ networks. Notice that the interpretation of the parallel relay network as a combination of two $X$ networks results is restrictive, since it implies that the relays are forced to decode, and hence restricts coding strategies to decode-and-forward based achievable schemes. However, in time-varying (or frequency-selective) channels, this scheme achieves the same degrees of freedom characterization as the amplify-and-forward based scheme in \cite{Boelcskei_Nabar_Oyman_Paulraj}.

\begin{theorem}
\label{thm:parallelrelay}
$\frac{MK}{2(M+K-1)}$ degrees of freedom are achievable in the two-hop parallel relay network with $M$ distributed transmitters and receivers with $K$ distributed relays. If $K \to \infty$, this parallel relay network has $M/2$ degree of freedom.
\end{theorem}

The proof follows from the degrees of freedom characterization of $X$ networks. For brevity's sake, we only provide an outline of the proof here. The message $W^{[j]}$ is split into $K$ independent sub-messages $W^{[j]}_k, k=1,2 \ldots K$ with message $W^{[j]}_k$ is meant to be decoded by relay $k$. In the first phase of duration $1$ time slot, the coding scheme corresponding to the $M \times K$ user $X$ network is employed so that $\frac{1}{M+K-1}$ bits corresponding to each sub-message is transmitted. For the second phase, notice that each relay has $\frac{1}{M+K-1}$ bits of information for each receiver - these bits are transmitted in a single time slot over the $K \times M$ $X$ network to the receivers. Since there is a total of $MK$ submessages in the system, a total of $\frac{MK}{M+K-1}$ bits are transmitted over the network 2 in time slots thus achieving $\frac{MK}{2(M+K-1)}$ degrees of freedom overall.

\section{Conclusion}
The $X$ network is arguably the most important single-hop network since it contains, within itself, most other one-way fully connected single hop networks. For instance, the 2 user MAC, BC and interference channels are all embedded in a two user $X$ channel, and therefore can be derived by setting appropriate messages to null.  We provide an outerbound for the degrees of freedom \emph{region} of the X channel with arbitrary number of single-antenna transmitters and receivers and no shared information among nodes. We also show that the \emph{total} number of degrees of freedom of the $M \times N$ $X$ network is equal to $\frac{MN}{M+N-1}$. The degrees of freedom region outerbound is very useful; it can be used to bound the number of degrees of freedom of most practical distributed single-hop wireless ad-hoc networks with single antenna at each node. 

It was observed that the $2$ user $X$ channel beats the $2$ user interference channel in performance with $4/3$ degrees of freedom. Recently, it has been shown that the $K$ user interference channel has $K/2$ degrees of freedom. The result of this paper implies that the $K$ user $X$ network is approximately equal to $K/2$ for large values of $K$. Therefore, for large values of $K$, $K$ user $X$ networks and $K$ user interference networks have approximately the same number of degrees of freedom.

In the $X$ network with $M$ transmitting and $N$ receiving nodes, if $M \gg N$ or if $N \gg M$, the total number of degrees of freedom is equivalent to the full cooperation MIMO outerbound of $\min(M,N)$ This is an optimistic result from the point of view of network information theory. It suggests that, from a degrees of freedo perspective, distributed single antenna nodes with no prior common information can behave as a single node with multiple antennas if they are transmitting to or receiving from a relatively large number of nodes. We provide an example of this scenario in the form of a parallel relay network. An interpretation of the two-hop parallel relay channel with $M$ transmitters and receivers and $K$ parallel relays as a composite of $2$ $X$ networks - one $M \times K$ and one $K \times M$ $X$ network - corresponding to the two hops leads to an alternate degrees of freedom optimal achievable scheme based on decode-and-forward. This example demonstrates the ubiquitousness of  $X$ networks in wireless communication networks. A study of $X$ networks with partial shared information among nodes can therefore potentially lead to understanding of several multi-hop networks.

The result of this work demonstrates the power of the technique of interference alignment combined with zero-forcing. The optimality of interference alignment in $X$ networks should motivate a closer look at interference alignment based schemes. For example, we note that the optimal achievable scheme uses arbitrarily long channel extensions in most cases. From a practical perspective, an important extension of this work is the study of the performance of achievable schemes restricted to fixed finite channel extensions. It has been observed that arbitrarily long channel extensions can be avoided without compromising optimality in the $3$ user interference channel if all the nodes are equipped with multiple antennas. A study of MIMO $X$ networks can potentially reveal more efficient schemes achieving optimality using shorter channel extensions. 

\appendices
\section{Proof of Lemma \ref{lemma:nonsingular}}
\label{app:nonsingular}
We need to show that $|\mathbf{A}|$, the determinant of $\mathbf{A}$ is non-zero with probability $1$.
Let $C_{ij}$ represents the co-factor corresponding to $a_{ij}$. Then, 
$$ |\mathbf{A}| = C_{11} a_{11} + C_{12} a_{12} \ldots C_{1N} a_{1N}$$
Notice that given all of $X_{jk}, j=2,3 \ldots M, k=1,2 \ldots K$, $|\mathbf{A}|$ is a polynomial in variables $X_{1k}, k=1,2, \ldots K$. Therefore, $|\mathbf{A}|=0$ with non-zero probability only if one of the following two conditions are satisfied.
\begin{enumerate}
\item The polynomial is the zero polynomial
\item $X_{1k}, k=1,2, \ldots K$ are roots of the polynomial formed by expanding the determinant
\end{enumerate}
If condition 1) is not satisfied, then the set of roots of the polynomial formed is a set of measure $0$. Notice that $C_{1l}$ is a polynomial in $X_{jk}, j=2,3, \ldots M, k=1,2 \ldots K$. Therefore, the variables $X_{1k}, k=1,2, \ldots , K$ have a continuous distribution given all of $C_{1l}$ and the probability that these variables take values from a set of measure $0$ is equal to $0$.  This clearly implies that the probability of condition $2$ being satisfied is equal to $0$. We now argue that the probability of condition $1$ being satisfied is also equal to $0$. 
Now, since each $a_{1j}$ is a \textit{unique} monomial term, condition $1$ is satisfied only if all the coefficients if $a_{1j}$ are zero, i.e., if $C_{1l}=0, l=1,2 \ldots M$. Therefore, 
$ \Pr(|\mathbf{A}| = 0) > 0 \Rightarrow \Pr(C_{1M}=0) > 0$
Note that $C_{1M}$ is formed by stripping the last row and last column of $\mathbf{A}$. Now, the same argument can be iteratively used, stripping the last row and last column at each stage, until we reach a single element matrix containing $a_{M1}$ i.e.
$$ \Pr(|\mathbf{A}| = 0) > 0 \Rightarrow \Pr(a_{M1} = 0) > 0$$
$a_{M1}$ is a product of terms of the form $ X_{Mk}, k=1,2, \ldots, K$ and therefore has a continuous probability distribution. We can hence conclude that $\Pr(a_{M1}=0) = 0$ and therefore $|\mathbf{A}|$ is non-zero almost surely. 
Thus, the column vectors of $\mathbf{A}$ are linearly independent with probability $1$.
\hfill\QED

\section{Proof of Lemma \ref{lemma:intalign}}
\label{app:intalign}
Let 
\begin{eqnarray*}
 \xV &=& \big\{ \big(\prod_{i=1,2 \ldots N} \xT_i^{\alpha_i}\big) \xw: (\alpha_1, \alpha_2, \ldots \alpha_N) \in \{1,2, \ldots n\} \big\} \\
 \xV^{'} &=& \big\{ \big(\prod_{i=1,2 \ldots N} \xT_i^{\alpha_i}\big) \xw : (\alpha_1, \alpha_2, \ldots \alpha_N) \in \{1,2, \ldots n+1\} \big\}
 \end{eqnarray*}
Note the slight abuse in notation in the above two equations - the right hand side represents the set of column vectors which may occur in any order to form the matrices on the corresponding left hand sides.
To clarify the notation, let $n=1$. Then the column vectors of $\xV$ is the column vector $\xT_1,\xT_2 \ldots \xT_N \xw$. $\xV^{'}$ contains vectors of the form $\xT_1^{\alpha_1} \xT_2^{\alpha_2} \ldots \xT_N^{\alpha_N} \xw$ for all $\alpha_i \in \{1,2\}$. 
It can be clearly seen that $\xV$ and $\xV^{'}$ satisfy 
$$\xT_i \xV \prec \xV^{'}, \forall i=1,2 \ldots N$$
Also, clearly, every entry in the $k^{th}$ row of $\xV$ and the $k^{th}$ row of $\xV^{'}$ is a monomial function in the variables of $T_{ji},j=1,2 \ldots N$ and $w_{i}$, where $T_{ji}$ is the diagonal entry in the $i$th row of $\xT_j$ and $w_{i}$ is a diagonal entry in the $i$th row of $\xw$.
All that needs to be shown is that matrices $\xV$ and $\xV^{'}$ have full rank almost surely. In other words, we need to show that their columns are linearly independent.
Notice that $\xV \prec \xV^{'}$. Therefore it is enough to show that the column vectors of $\xV^{'}$ are linearly independent with probability $1$. Now consider the matrix $\mathbf{\Lambda}$ as below
$$ \mathbf{\Lambda} = \left[ \mathbf{V}^{'} \mathbf{U} \right]$$
where $\mathbf{U}$ is a $\mu \times (\mu - (n+1)^N)$ matrix whose entries $u_{ij}$ are chosen i.i.d from a continuous distribution. Now, notice that the entries in row $i$ of $\mathbf{\Lambda}$ are unique monomial terms in variables $w_i, T_{j,i}, j=1,2 \ldots N, u_{ik}, k=1,2 \ldots (\mu-(n+1)^N)$. Now, we invoke Lemma \ref{lemma:nonsingular} to conclude that the column vectors of $\mathbf{\Lambda}$ are linearly independent with probability $1$. Since the set of column vectors of $\xV$ and those of $\xV^{'}$ are contained in $\mathbf{\Lambda}$, we can conclude that the column vectors of $\xV$ and $\xV^{'}$ are linearly independent with probability $1$.
\QED

\section{Degrees of freedom of the General $M\times N$ user $X$ Network}
\label{app:proof_dofx}
The achievability scheme generalizes the scheme used for the $K$ user interference channel in \cite{Cadambe_Jafar_int}.
Let $\Gamma=(M-1)(N-1)$. We will develop a coding scheme based on interference alignment which achieves a total of $\frac{N(n+1)^{\Gamma} + (M-1)Nn^{\Gamma}}{N(n+1)^{\Gamma}+(M-1)n^{\Gamma}}$ degrees of freedom for any arbitrary $n$. Taking supremum over all $n$ proves that the total number of degrees of freedom is equal to $\frac{MN}{M+N-1}$ as desired. To show this, we construct a scheme that achieves a total of $ (M-1)Nn^{\Gamma}+N(n+1)^{\Gamma}$ degrees of freedom over a $\mu_n=N(n+1)^{\Gamma}+(M-1)n^{\Gamma}$ symbol extension of the original channel. Over the extended channel, the scheme achieves $(n+1)^{\Gamma}$ degrees of freedom for each of the $N$ messages $W^{[j1]}, j=1,2 \ldots N$ and achieves $n^{\Gamma}$ degrees of freedom for all the other $(M-1)N$ messages $W^{[ji]}, j=1,2 \ldots N, i=2,3 \ldots N$.
The signal vector in the extended channel at the $j^{th}$ user's receiver can be expressed as 
$$ \xY^{[j]}(\kappa) = \displaystyle\sum_{i=1}^{M} \xH^{[ji]}(\kappa) \mathbf{X}^{[i]}(\kappa)$$ 
where $\mathbf{X}^{[i]}$ is a $\mu_n \times 1$ column vector representing the $\mu_n$ symbol extension of the transmitted symbol $X^{[i]}$, i.e 
$$\xX^{[i]}(\kappa) \define \left[ \begin{array}{c} X^{[i]}(\mu_n\kappa+1) \\X^{[i]}(\mu_n \kappa + 2)\\ \vdots \\ X^{[i]}(\mu_n (\kappa+1)) \end{array}\right]$$
Similarly $\xY^{[i]}$ and $\xZ^{[i]}$ represent $\mu_n$ symbol extensions of the $Y^{[i]}$ and $Z^{[i]}$ respectively. 
$\xH^{[ji]}$ is a diagonal $\mu_n\times \mu_n$ matrix representing the $\mu_n$ symbol extension of the channel i.e
$$ \xH^{[ji]}(\kappa) \define \left[ \begin{array}{cccc}  H^{[ji]}(\mu_n \kappa + 1) & 0 & \ldots & 0\\
	0 & H^{[ji]}(\mu_n \kappa+2) & \ldots & 0\\
	\vdots & \cdots & \ddots & \vdots\\ 
	0 & 0& \cdots  & H^{[ji]}(\mu_n (\kappa+1)) \end{array}\right] $$
Over the extended channel, message $W_{j1}$ is encoded at transmitter $1$ as $(n+1)^{\Gamma}$ independent streams $x^{[j1]}_m, m=1,2, \ldots (n+1)^{\Gamma}$ along directions $\mathbf{v}_m^{[j1]},m=1,2 \ldots (n+1)^{\Gamma}$. So the signal transmitted at transmitter $1$ may be written as 
$$\xX^{[1]}(t) = \displaystyle\sum_{j=1}^{N}\displaystyle\sum_{m=1}^{(n+1)^\Gamma} x^{[j1]}_m(t) \mathbf{v}_m^{[j1]}(t) = \displaystyle\sum_{j=1}^{N}\xV^{[j1]}(t) \mathbf{x}^{[j1]}(t)$$
Note that $\xV^{[j1]}$ is a $\mu_n \times (n+1)^\Gamma$ matrix and $\mathbf{x}^{[j1]}(t)$ is a $(n+1)^{\Gamma} \times 1$ vector.
Similarly $W_{ji}, i\neq 1$ is encoded into $n^{\Gamma}$ independent streams by transmitter $i$ as 
$$\xX^{[i]}(t) = \displaystyle\sum_{j=1}^{N}\displaystyle\sum_{m=1}^{n^\Gamma} x^{[ji]}_m(t) \mathbf{v}_m^{[ji]}(t) = \displaystyle\sum_{j=1}^{N}\xV^{[ji]}(t) \mathbf{x}^{[ji]}(t)$$

where $\xV^{[ji]}$ is a $\mu_n\times n^\Gamma$ matrix
The received signal at the $k^{th}$ receiver can then be written as
$$ \xY^{[k]}(t) = \displaystyle\sum_{i=1}^{M}\xH^{[ki]}(t) \big( \displaystyle\sum_{j=1}^{N}\xV^{[ji]}(t) \mathbf{x}^{[ji]}(t) \big) $$

We wish to design beamforming directions  $\xV^{[ki]}$ so that receiver $j$ can decode each of the desired signals by zeroforcing the interference. We ensure this by aligning interference so that the dimension of the space spanned by the interference vectors at any receiver is equal to $(N-1)(n+1)^\Gamma$. Once the interference is aligned in this fashion, a receiver can decode its desired $(n+1)^\Gamma+(M-1)n^\Gamma$ streams by zero-forcing the interference in the $\mu_n = (N-1)(n+1)^\Gamma+(n+1)^\Gamma+(M-1)n^\Gamma$ dimensional space. 
Interference alignment is ensured by choosing the beamforming directions $\xV^{[ki]}$ so that the following interference alignment equations are satisfied at receiver $j, \forall j=1,2 \ldots N$.
\begin{eqnarray} 
 \left.
	 \begin{array}{ccc}
\xH^{[ji]} \xV^{[1i]} &\prec& \xH^{[j1]} \xV^{[11]} \\ 
\xH^{[ji]} \xV^{[2i]} &\prec& \xH^{[j1]} \xV^{[21]} \\
& \vdots & \\
\xH^{[ji]} \xV^{[(j-1)i]} &\prec& \xH^{[j1]} \xV^{[(j-1)1]}\\
\xH^{[ji]} \xV^{[(j+1)i]} &\prec& \xH^{[j1]} \xV^{[(j+1)1]} \\
& \vdots & \\
\xH^{[ji]} \xV^{[Ni]} &\prec& \xH^{[j1]} \xV^{[N1]}
    \end{array} \right\} &  \forall i=2,3, \ldots M & 
	\label{eqn:subspace_1}
	\end{eqnarray} 
where $\mathbf{P} \prec \mathbf{Q}$, means that the set of column vectors of matrix $\mathbf{P}$ is a subset of the set of column vectors of matrix $\mathbf{Q}$.
In other words, we wish to pick matrices $\xV^{[ki]}$ so that, at receiver $j$, all the interfering spaces from transmitters $2, 3, \ldots M$ align themselves with the interference from transmitter $1$. Then, at any receiver, the dimension of \emph{all} the interfering streams is equal to the dimension of the interference from transmitter $1$ which is equal to $(N-1)(n+1)^{\Gamma}$ as required. Note that there are $(M-1)(N-1) = \Gamma$ relations above corresponding to receiver $j$. Therefore a total of $\Gamma N$ relations of the form $\mathbf{P} \prec \mathbf{Q}$ need to be satisfied. These relations can be expressed alternately as 

\begin{eqnarray} 
 \left.
	 \begin{array}{ccc}
\xH^{[1i]} \xV^{[ki]} &\prec& \xH^{[11]} \xV^{[k1]}  \\ 
\xH^{[2i]} \xV^{[ki]} &\prec& \xH^{[21]} \xV^{[k1]} \\ 
& \vdots &\\
\xH^{[(k-1)i]} \xV^{[ki]} &\prec& \xH^{[(k-1)1]} \xV^{[k1]} \\ 
\xH^{[(k+1)i]} \xV^{[ki]} &\prec& \xH^{[(k+1)1]} \xV^{[k1]}  \\ 
& \vdots &\\
\xH^{[Ni]} \xV^{[ki]} &\prec& \xH^{[N1]} \xV^{[k1]} \\ 
    \end{array}
		 \right\} &  \forall i=2,3, \ldots M,\forall k=1,2 \ldots N  & 
	\label{eqn:subspace_2}
\end{eqnarray}

In order to satisfy the above relations, we first choose 
\begin{equation*}  
\xV^{[k2]} = \xV^{[k3]} = \ldots \xV^{[kM]}, \forall k=1,2 \ldots N
\end{equation*}

Now, the relations in (\ref{eqn:subspace_2}) can be re-written as 
\begin{eqnarray}
 \left.
	 \begin{array}{ccc}
\xT^{[1i]} \xV^{[k2]} &\prec& \xV^{[k1]}   \\ 
\xT^{[2i]} \xV^{[k2]} &\prec&  \xV^{[k1]}  \\ 
\xT^{[3i]} \xV^{[k2]} &\prec&  \xV^{[k1]}  \\ 
& \vdots &\nonumber\\
\xT^{[(k-1)i]} \xV^{[k2]} &\prec& \xV^{[k1]} \\ 
\xT^{[(k+1)i]} \xV^{[k2]} &\prec& \xV^{[k1]} \\ 
& \vdots &\\
\xT^{[Ni]} \xV^{[k2]} &\prec& \xV^{[k1]} 
    \end{array} \right\} &  \forall i=2,3 \ldots M, k=1,2 \ldots N& \\
	\label{eqn:subspace_3}
\end{eqnarray}

where
\begin{eqnarray*}
 \xT^{[ji]} = (\xH^{[j1]})^{-1} \xH^{[ji]}, & j=1,2 \ldots N, i=2,3 \ldots M \nonumber & 
\end{eqnarray*}
We now wish to pick $\xV^{[k1]}$ and $\xV^{[k2]}$ so that the above relations are satisfied and then use $\xV^{[ki]} = \xV^{[k2]}, i=3,4 \ldots M$.
To satisfy the above relations, we first generate $\mu_n \times 1$ column vectors $\xw^{[k]}, k=1,2 \ldots N$ such that all the entries of the matrix $\left[ \xw^{[1]}~ \xw^{[2]}~ \ldots~ \xw^{[N]}\right]$ are chosen i.i.d from \emph{some} continuous distribution whose support lies between a finite minimum value and a finite maximum value.
Notice that for a fixed $k$, there are $\Gamma$ interference alignment relations of the form in Lemma \ref{lemma:intalign}. Using column vector $\xw^{[k]}$ (which has non-zero entries with probability $1$), the construction of the Lemma can be used to construct vector spaces $\xV^{[k2]}$ and $\xV^{[k1]}$ satisfying the desired interference alignment relations of (\ref{eqn:subspace_3}). Also, the construction ensures that $\mbox{rank}(\xV^{[k2]}) = n^{\Gamma}$ and $\mbox{rank}(\xV^{[k1]}) = (n+1)^\Gamma$ as required.

Now, we have designed $\xV^{[ji]}$ so that the desired interference alignment equations of (\ref{eqn:subspace_1}) are satisfied. We now need to ensure that at each receiver, all the desired signal streams are linearly independent of each other and independent of the interference, so that they can be decoded using zero-forcing . Notice that at any receiver $k$, interference alignment ensures that all the interference vectors arrive along $\xH^{[k1]} \xV^{[j1]},j=1,2 \ldots k-1,k+1, \ldots N$ and therefore,  the interference space is the space spanned by the $(N-1)(n+1)^\Gamma$ column vectors of $\mathbf{I}_k$ where 
$$ \mathbf{I}_k = \left[ \xH^{[k1]} \xV^{[11]}~~\xH^{[k1]} \xV^{[21]}~~\ldots~~\xH^{[k1]} \xV^{[(k-1)1]} ~~ \xH^{[k1]} \xV^{[k+1]1}~~ \ldots \xH^{[k1]} \xV^{[N1]} \right] $$ 
The desired streams at receiver $k$ arrive along the $(n+1)^\Gamma + (M-1)n^\Gamma$ column vectors of $\mathbf{D}_k$ where 
\begin{eqnarray*}
 \mathbf{D}_k &=& \left[ \xH^{[k1]} \xV^{[k1]}~~\xH^{[k2]} \xV^{[k2]} ~~ \ldots \xH^{[kM]} \xV^{[kM]} \right] \\
 &=& \left[ \xH^{[k1]} \xV^{[k1]}~~\xH^{[k2]} \xV^{[k2]} ~~ \ldots \xH^{[kM]} \xV^{[k2]} \right]
\end{eqnarray*}
So, at receiver $k$, we need to ensure that the matrix 
$$ \mathbf{\Lambda_k} = \left[ \mathbf{D}_k ~~\mathbf{I}_k \right]$$
has a full rank of $\mu_n$ almost surely.  Now, notice that an element in the $m$th row of $\mathbf{\Lambda_k}$ is a monomial term in $H^{[ji]}_m$ and $w^{[j]}_m$ for $i=1,2 \ldots M, j=1,2 \ldots N,$, where $H^{[ji]}_m$ represents the diagonal entry in the $m$th row of $\xH^{[ji]}$ and $w^{[j]}_m$ represents the entry in the $m$th row of the column vector $\xw^{[j]}$. We intend to use Lemma \ref{lemma:nonsingular} to show that the matrix $\mathbf{\Lambda}$ has full rank with probability $1$. To do so, we need to verify that each of the monomial terms in row $m$ are unique. We now make the following observations
\begin{enumerate}
\item In the $m^{th}$ row, the monomial entries from $\xH^{[ki]} \xV^{[ji]}$ contain $w^{[j]}_m$ with exponent 1, but do not contain $w^{[l]}_m, l \neq j$
\item Notice that the equation corresponding to $\xH^{[ki]}, i=2,3 \ldots M$ is missing in the interference alignment relations of (\ref{eqn:subspace_2}) at receiver $k$. The construction of Lemma \ref{lemma:intalign} ensures that monomial entries in the $m^{th}$ row of $\xV^{[k2]}$ do not contain $H^{[ki]}_m, i=2 \ldots M$.
\end{enumerate}
Observation 1) implies that all the monomial entries of the $m$th row of $\mathbf{I}_k$ are unique. Furthermore, it also implies that all the monomial terms in $\mathbf{I}_k$ are different from all the monomials in $\mathbf{D}_k$. Now, observation 2) implies that all the entries in $\mathbf{D}_k$ are unique, since the term $H^{[ki]}_m$ occurs only in the column vectors corresponding to $\xH^{[ki]} \xV^{[k2]}$ . Therefore, all the monomial entries of $\mathbf{\Lambda}$ are unique and from the result of Lemma \ref{lemma:nonsingular}, we can conclude that all the entries in the $m$th row of $\mathbf{\Lambda}_k$ are unique monomial terms and therefore, the matrix has a full rank of $\mu_n$ almost surely.

Thus, the desired signal is linearly independent of all the receivers and therefore, using the techniques of interference alignment and zero-forcing, $\frac{N(n+1)^{\Gamma} + (M-1)Nn^{\Gamma}}{N(n+1)^{\Gamma}+(M-1)n^{\Gamma}}$ degrees of freedom are achievable over the $M \times N$ user $X$ network for any $n \in \mathbb{N}$. Taking supremum over $n$, we conclude that the $M\times N$ user $X$ network has $\frac{MN}{M+N-1}$ degrees of freedom.

\bibliographystyle{ieeetr}
\bibliography{Thesis}
\end{document}

%% file: X2tx3rx.eepic
\setlength{\unitlength}{0.00045833in}
\begingroup\makeatletter\ifx\SetFigFont\undefined%
\gdef\SetFigFont#1#2#3#4#5{%
  \reset@font\fontsize{#1}{#2pt}%
  \fontfamily{#3}\fontseries{#4}\fontshape{#5}%
  \selectfont}%
\fi\endgroup%
{\renewcommand{\dashlinestretch}{30}
\begin{picture}(6924,4089)(0,-10)
\put(312,3387){\makebox(0,0)[lb]{{\SetFigFont{7}{8.4}{\familydefault}{\mddefault}{\updefault}$W^{[21]}$}}}
\put(312,3012){\makebox(0,0)[lb]{{\SetFigFont{7}{8.4}{\familydefault}{\mddefault}{\updefault}$W^{[31]}$}}}
\put(312,3762){\makebox(0,0)[lb]{{\SetFigFont{7}{8.4}{\familydefault}{\mddefault}{\updefault}$W^{[11]}$}}}
\put(237,1437){\makebox(0,0)[lb]{{\SetFigFont{7}{8.4}{\familydefault}{\mddefault}{\updefault}$W^{[32]}$}}}
\put(237,1812){\makebox(0,0)[lb]{{\SetFigFont{7}{8.4}{\familydefault}{\mddefault}{\updefault}$W^{[22]}$}}}
\put(237,2187){\makebox(0,0)[lb]{{\SetFigFont{7}{8.4}{\familydefault}{\mddefault}{\updefault}$W^{[12]}$}}}
\path(2262,3462)(4962,3462)
\path(2262,1962)(4962,1962)
\path(2262,1962)(4962,462)
\path(2262,1962)(4962,3462)
\path(2262,3462)(4962,1962)
\path(2262,3462)(4962,462)
\path(12,2487)(1137,2487)(1137,1362)
	(12,1362)(12,2487)
\path(1287,1962)(1662,1962)
\path(1542.000,1932.000)(1662.000,1962.000)(1542.000,1992.000)
\path(1212,3462)(1587,3462)
\path(1467.000,3432.000)(1587.000,3462.000)(1467.000,3492.000)
\path(87,4062)(1212,4062)(1212,2937)
	(87,2937)(87,4062)
\path(5562,1962)(5937,1962)
\path(5817.000,1932.000)(5937.000,1962.000)(5817.000,1992.000)
\path(5562,3462)(5937,3462)
\path(5817.000,3432.000)(5937.000,3462.000)(5817.000,3492.000)
\path(5487,462)(5862,462)
\path(5742.000,432.000)(5862.000,462.000)(5742.000,492.000)
\path(5862,3762)(6837,3762)(6837,2937)
	(5862,2937)(5862,3762)
\path(5937,2337)(6912,2337)(6912,1512)
	(5937,1512)(5937,2337)
\path(5862,837)(6837,837)(6837,12)
	(5862,12)(5862,837)
\put(5037,3387){\makebox(0,0)[lb]{{\SetFigFont{7}{8.4}{\rmdefault}{\mddefault}{\updefault}$Y^{[1]}$}}}
\put(5037,1887){\makebox(0,0)[lb]{{\SetFigFont{7}{8.4}{\rmdefault}{\mddefault}{\updefault}$Y^{[2]}$}}}
\put(5037,387){\makebox(0,0)[lb]{{\SetFigFont{7}{8.4}{\rmdefault}{\mddefault}{\updefault}$Y^{[3]}$}}}
\put(1662,3387){\makebox(0,0)[lb]{{\SetFigFont{7}{8.4}{\rmdefault}{\mddefault}{\updefault}$X^{[1]}$}}}
\put(1662,1887){\makebox(0,0)[lb]{{\SetFigFont{7}{8.4}{\rmdefault}{\mddefault}{\updefault}$X^{[2]}$}}}
\put(6012,3537){\makebox(0,0)[lb]{{\SetFigFont{7}{8.4}{\familydefault}{\mddefault}{\updefault}$\widehat{W}^{[11]}$}}}
\put(6012,3162){\makebox(0,0)[lb]{{\SetFigFont{7}{8.4}{\familydefault}{\mddefault}{\updefault}$\widehat{W}^{[12]}$}}}
\put(6087,2037){\makebox(0,0)[lb]{{\SetFigFont{7}{8.4}{\familydefault}{\mddefault}{\updefault}$\widehat{W}^{[21]}$}}}
\put(6087,1737){\makebox(0,0)[lb]{{\SetFigFont{7}{8.4}{\familydefault}{\mddefault}{\updefault}$\widehat{W}^{[22]}$}}}
\put(6087,237){\makebox(0,0)[lb]{{\SetFigFont{7}{8.4}{\familydefault}{\mddefault}{\updefault}$\widehat{W}^{[31]}$}}}
\put(6087,537){\makebox(0,0)[lb]{{\SetFigFont{7}{8.4}{\familydefault}{\mddefault}{\updefault}$\widehat{W}^{[32]}$}}}
\end{picture}
}

%% file: Xactivepairs.eepic
\setlength{\unitlength}{0.00036667in}
\begingroup\makeatletter\ifx\SetFigFont\undefined%
\gdef\SetFigFont#1#2#3#4#5{%
  \reset@font\fontsize{#1}{#2pt}%
  \fontfamily{#3}\fontseries{#4}\fontshape{#5}%
  \selectfont}%
\fi\endgroup%
{\renewcommand{\dashlinestretch}{30}
\begin{picture}(17990,5853)(0,-10)
\put(9087,4638){\makebox(0,0)[lb]{{\SetFigFont{5}{6.0}{\familydefault}{\mddefault}{\updefault}$W^{[21]}$}}}
\put(9087,4263){\makebox(0,0)[lb]{{\SetFigFont{5}{6.0}{\familydefault}{\mddefault}{\updefault}$W^{[31]}$}}}
\put(9087,5013){\makebox(0,0)[lb]{{\SetFigFont{5}{6.0}{\familydefault}{\mddefault}{\updefault}$W^{[11]}$}}}
\put(237,4563){\makebox(0,0)[lb]{{\SetFigFont{5}{6.0}{\familydefault}{\mddefault}{\updefault}$W^{[21]}$}}}
\put(237,4188){\makebox(0,0)[lb]{{\SetFigFont{5}{6.0}{\familydefault}{\mddefault}{\updefault}$W^{[31]}$}}}
\put(237,4938){\makebox(0,0)[lb]{{\SetFigFont{5}{6.0}{\familydefault}{\mddefault}{\updefault}$W^{[11]}$}}}
\dashline{60.000}(16887,2463)(14487,2463)(14487,3213)
\path(14517.000,3093.000)(14487.000,3213.000)(14457.000,3093.000)
\dashline{60.000}(14487,2463)(14487,1713)
\path(14457.000,1833.000)(14487.000,1713.000)(14517.000,1833.000)
\path(11037,4713)(13737,4713)
\path(11037,3213)(13737,4713)
\path(11037,4713)(13737,3213)
\path(11037,4713)(13737,1713)
\path(9987,4713)(10362,4713)
\path(10242.000,4683.000)(10362.000,4713.000)(10242.000,4743.000)
\path(8862,5313)(9987,5313)(9987,4188)
	(8862,4188)(8862,5313)
\path(14337,4713)(14712,4713)
\path(14592.000,4683.000)(14712.000,4713.000)(14592.000,4743.000)
\path(9987,3213)(10362,3213)
\path(10242.000,3183.000)(10362.000,3213.000)(10242.000,3243.000)
\path(14412,3213)(14787,3213)
\path(14667.000,3183.000)(14787.000,3213.000)(14667.000,3243.000)
\path(8937,3438)(9912,3438)(9912,2988)
	(8937,2988)(8937,3438)
\path(14337,1713)(14712,1713)
\path(14592.000,1683.000)(14712.000,1713.000)(14592.000,1743.000)
\path(13907,5609)(13757,4784)
\path(13748.950,4907.431)(13757.000,4784.000)(13807.982,4896.698)
\path(14787,3438)(15612,3438)(15612,3063)
	(14787,3063)(14787,3438)
\path(14712,5388)(15687,5388)(15687,3888)
	(14712,3888)(14712,5388)
\path(14712,1938)(15537,1938)(15537,1563)
	(14712,1563)(14712,1938)
\path(2187,4638)(4887,4638)
\path(2187,3138)(4887,3138)
\path(2187,3138)(4887,1638)
\path(2187,3138)(4887,4638)
\path(2187,4638)(4887,3138)
\path(2187,4638)(4887,1638)
\path(1137,4638)(1512,4638)
\path(1392.000,4608.000)(1512.000,4638.000)(1392.000,4668.000)
\path(12,5238)(1137,5238)(1137,4113)
	(12,4113)(12,5238)
\path(5487,4638)(5862,4638)
\path(5742.000,4608.000)(5862.000,4638.000)(5742.000,4668.000)
\path(1137,3138)(1512,3138)
\path(1392.000,3108.000)(1512.000,3138.000)(1392.000,3168.000)
\path(5562,3138)(5937,3138)
\path(5817.000,3108.000)(5937.000,3138.000)(5817.000,3168.000)
\path(87,3363)(1062,3363)(1062,2913)
	(87,2913)(87,3363)
\path(5487,1638)(5862,1638)
\path(5742.000,1608.000)(5862.000,1638.000)(5742.000,1668.000)
\path(5937,3363)(6762,3363)(6762,2988)
	(5937,2988)(5937,3363)
\path(5862,1863)(6687,1863)(6687,1488)
	(5862,1488)(5862,1863)
\path(5937,5013)(6762,5013)(6762,4188)
	(5937,4188)(5937,5013)
\put(17037,2388){\makebox(0,0)[lb]{{\SetFigFont{5}{6.0}{\rmdefault}{\mddefault}{\updefault}$W^{[12]}$}}}
\put(15462,2538){\makebox(0,0)[lb]{{\SetFigFont{5}{6.0}{\familydefault}{\mddefault}{\updefault}Genie}}}
\put(12012,63){\makebox(0,0)[lb]{{\SetFigFont{6}{7.2}{\sfdefault}{\mddefault}{\updefault}               (b)}}}
\put(3012,63){\makebox(0,0)[lb]{{\SetFigFont{6}{7.2}{\sfdefault}{\mddefault}{\updefault}               (a)}}}
\put(13812,4638){\makebox(0,0)[lb]{{\SetFigFont{5}{6.0}{\rmdefault}{\mddefault}{\updefault}$Y^{[1]}$}}}
\put(13812,3138){\makebox(0,0)[lb]{{\SetFigFont{5}{6.0}{\rmdefault}{\mddefault}{\updefault}$Y^{[2]}$}}}
\put(13812,1638){\makebox(0,0)[lb]{{\SetFigFont{5}{6.0}{\rmdefault}{\mddefault}{\updefault}$Y^{[3]}$}}}
\put(10437,4638){\makebox(0,0)[lb]{{\SetFigFont{5}{6.0}{\rmdefault}{\mddefault}{\updefault}$X^{[1]}$}}}
\put(10437,3138){\makebox(0,0)[lb]{{\SetFigFont{5}{6.0}{\rmdefault}{\mddefault}{\updefault}$X^{[2]}$}}}
\put(11412,738){\makebox(0,0)[lb]{{\SetFigFont{5}{6.0}{\familydefault}{\mddefault}{\updefault}$W^{[22]}=\phi$}}}
\put(11412,1038){\makebox(0,0)[lb]{{\SetFigFont{5}{6.0}{\familydefault}{\mddefault}{\updefault}$W^{[32]}=\phi$}}}
\put(9012,3138){\makebox(0,0)[lb]{{\SetFigFont{5}{6.0}{\familydefault}{\mddefault}{\updefault}$W^{[12]}$}}}
\put(14862,4413){\makebox(0,0)[lb]{{\SetFigFont{5}{6.0}{\familydefault}{\mddefault}{\updefault}$\widehat{W}^{[21]}$}}}
\put(14787,1638){\makebox(0,0)[lb]{{\SetFigFont{5}{6.0}{\familydefault}{\mddefault}{\updefault}$\widehat{W}^{[31]}$}}}
\put(14862,3138){\makebox(0,0)[lb]{{\SetFigFont{5}{6.0}{\familydefault}{\mddefault}{\updefault}$\widehat{W}^{[21]}$}}}
\put(14862,4113){\makebox(0,0)[lb]{{\SetFigFont{5}{6.0}{\familydefault}{\mddefault}{\updefault}$\widehat{W}^{[31]}$}}}
\put(14862,5088){\makebox(0,0)[lb]{{\SetFigFont{5}{6.0}{\familydefault}{\mddefault}{\updefault}$\widehat{W}^{[11]}$}}}
\put(13812,5688){\makebox(0,0)[lb]{{\SetFigFont{6}{7.2}{\rmdefault}{\mddefault}{\updefault}Reduce noise}}}
\put(4962,4563){\makebox(0,0)[lb]{{\SetFigFont{5}{6.0}{\rmdefault}{\mddefault}{\updefault}$Y^{[1]}$}}}
\put(4962,3063){\makebox(0,0)[lb]{{\SetFigFont{5}{6.0}{\rmdefault}{\mddefault}{\updefault}$Y^{[2]}$}}}
\put(4962,1563){\makebox(0,0)[lb]{{\SetFigFont{5}{6.0}{\rmdefault}{\mddefault}{\updefault}$Y^{[3]}$}}}
\put(1587,4563){\makebox(0,0)[lb]{{\SetFigFont{5}{6.0}{\rmdefault}{\mddefault}{\updefault}$X^{[1]}$}}}
\put(1587,3063){\makebox(0,0)[lb]{{\SetFigFont{5}{6.0}{\rmdefault}{\mddefault}{\updefault}$X^{[2]}$}}}
\put(2562,663){\makebox(0,0)[lb]{{\SetFigFont{5}{6.0}{\familydefault}{\mddefault}{\updefault}$W^{[22]}=\phi$}}}
\put(2562,963){\makebox(0,0)[lb]{{\SetFigFont{5}{6.0}{\familydefault}{\mddefault}{\updefault}$W^{[32]}=\phi$}}}
\put(162,3063){\makebox(0,0)[lb]{{\SetFigFont{5}{6.0}{\familydefault}{\mddefault}{\updefault}$W^{[12]}$}}}
\put(6012,4713){\makebox(0,0)[lb]{{\SetFigFont{5}{6.0}{\familydefault}{\mddefault}{\updefault}$\widehat{W}^{[11]}$}}}
\put(6012,4338){\makebox(0,0)[lb]{{\SetFigFont{5}{6.0}{\familydefault}{\mddefault}{\updefault}$\widehat{W}^{[12]}$}}}
\put(5937,1563){\makebox(0,0)[lb]{{\SetFigFont{5}{6.0}{\familydefault}{\mddefault}{\updefault}$\widehat{W}^{[31]}$}}}
\put(6012,3063){\makebox(0,0)[lb]{{\SetFigFont{5}{6.0}{\familydefault}{\mddefault}{\updefault}$\widehat{W}^{[21]}$}}}
\put(14862,4788){\makebox(0,0)[lb]{{\SetFigFont{5}{6.0}{\familydefault}{\mddefault}{\updefault}$\widehat{W}^{[12]}$}}}
\end{picture}
}

%% file: 2usrXdelay.eepic
\setlength{\unitlength}{0.00050000in}
\begingroup\makeatletter\ifx\SetFigFont\undefined%
\gdef\SetFigFont#1#2#3#4#5{%
  \reset@font\fontsize{#1}{#2pt}%
  \fontfamily{#3}\fontseries{#4}\fontshape{#5}%
  \selectfont}%
\fi\endgroup%
{\renewcommand{\dashlinestretch}{30}
\begin{picture}(12699,4584)(0,-10)
\path(11487,3600)(12087,3600)(12087,3150)
	(11487,3150)(11487,3600)
\path(10887,750)(11487,750)(11487,300)
	(10887,300)(10887,750)
\path(11487,750)(12087,750)(12087,300)
	(11487,300)(11487,750)
\path(8487,750)(9087,750)(9087,300)
	(8487,300)(8487,750)
\path(9687,750)(10287,750)(10287,300)
	(9687,300)(9687,750)
\path(12687,300)(12087,300)(12087,750)(12687,750)
\path(10287,750)(10887,750)(10887,300)
	(10287,300)(10287,750)
\path(9087,750)(9687,750)(9687,300)
	(9087,300)(9087,750)
\path(10887,3600)(11487,3600)(11487,3150)
	(10887,3150)(10887,3600)
\path(8487,3600)(9087,3600)(9087,3150)
	(8487,3150)(8487,3600)
\path(9687,3600)(10287,3600)(10287,3150)
	(9687,3150)(9687,3600)
\path(12687,3150)(12087,3150)(12087,3600)(12687,3600)
\path(10287,3600)(10887,3600)(10887,3150)
	(10287,3150)(10287,3600)
\path(9087,3600)(9687,3600)(9687,3150)
	(9087,3150)(9087,3600)
\path(8869.096,917.521)(8787.000,825.000)(8902.977,868.003)
\path(8787,825)(10212,1800)(10587,825)
\path(10515.922,926.232)(10587.000,825.000)(10571.923,947.771)
\path(9487.415,3002.772)(9387.000,3075.000)(9441.611,2964.015)
\path(9387,3075)(10212,2100)(11112,3075)
\path(11052.650,2966.475)(11112.000,3075.000)(11008.562,3007.172)
\put(11487,3300){\makebox(0,0)[lb]{{\SetFigFont{6}{7.2}{\rmdefault}{\mddefault}{\updefault}   $W^{[12]}$}}}
\put(11487,450){\makebox(0,0)[lb]{{\SetFigFont{6}{7.2}{\rmdefault}{\mddefault}{\updefault}   $W^{[22]}$}}}
\put(12237,450){\makebox(0,0)[lb]{{\SetFigFont{6}{7.2}{\rmdefault}{\mddefault}{\updefault}  $\ldots$}}}
\put(9162,450){\makebox(0,0)[lb]{{\SetFigFont{6}{7.2}{\rmdefault}{\mddefault}{\updefault}$W^{[21]}$}}}
\put(9762,450){\makebox(0,0)[lb]{{\SetFigFont{6}{7.2}{\rmdefault}{\mddefault}{\updefault}$W^{[22]}$}}}
\put(10362,375){\makebox(0,0)[lb]{{\SetFigFont{5}{6.0}{\rmdefault}{\mddefault}{\updefault}$W^{[12]}$}}}
\put(10287,525){\makebox(0,0)[lb]{{\SetFigFont{5}{6.0}{\rmdefault}{\mddefault}{\updefault} $W^{[11]}$}}}
\put(9162,3225){\makebox(0,0)[lb]{{\SetFigFont{5}{6.0}{\rmdefault}{\mddefault}{\updefault}$W^{[22]}$}}}
\put(8487,3300){\makebox(0,0)[lb]{{\SetFigFont{6}{7.2}{\rmdefault}{\mddefault}{\updefault} $W^{[11]}$}}}
\put(10962,3225){\makebox(0,0)[lb]{{\SetFigFont{5}{6.0}{\rmdefault}{\mddefault}{\updefault}$W^{[22]}$}}}
\put(9687,3300){\makebox(0,0)[lb]{{\SetFigFont{6}{7.2}{\rmdefault}{\mddefault}{\updefault}   $W^{[12]}$}}}
\put(10887,3375){\makebox(0,0)[lb]{{\SetFigFont{5}{6.0}{\rmdefault}{\mddefault}{\updefault} $W^{[21]}$}}}
\put(9162,3375){\makebox(0,0)[lb]{{\SetFigFont{5}{6.0}{\rmdefault}{\mddefault}{\updefault}$W^{[21]}$}}}
\put(10362,3300){\makebox(0,0)[lb]{{\SetFigFont{6}{7.2}{\rmdefault}{\mddefault}{\updefault}$W^{[11]}$}}}
\put(12237,3300){\makebox(0,0)[lb]{{\SetFigFont{6}{7.2}{\rmdefault}{\mddefault}{\updefault}  $\ldots$}}}
\put(8562,525){\makebox(0,0)[lb]{{\SetFigFont{5}{6.0}{\rmdefault}{\mddefault}{\updefault}$W^{[11]}$}}}
\put(10962,450){\makebox(0,0)[lb]{{\SetFigFont{6}{7.2}{\rmdefault}{\mddefault}{\updefault}$W^{[21]}$}}}
\put(8562,375){\makebox(0,0)[lb]{{\SetFigFont{5}{6.0}{\rmdefault}{\mddefault}{\updefault}$W^{[12]}$}}}
\put(9612,1875){\makebox(0,0)[lb]{{\SetFigFont{7}{8.4}{\rmdefault}{\mddefault}{\updefault}Interference Alignment}}}
\put(8487,2925){\makebox(0,0)[lb]{{\SetFigFont{7}{8.4}{\rmdefault}{\mddefault}{\updefault}0}}}
\put(9087,2925){\makebox(0,0)[lb]{{\SetFigFont{7}{8.4}{\rmdefault}{\mddefault}{\updefault}1}}}
\put(11487,2925){\makebox(0,0)[lb]{{\SetFigFont{7}{8.4}{\rmdefault}{\mddefault}{\updefault}2}}}
\put(9687,2925){\makebox(0,0)[lb]{{\SetFigFont{7}{8.4}{\rmdefault}{\mddefault}{\updefault}2}}}
\put(12087,2925){\makebox(0,0)[lb]{{\SetFigFont{7}{8.4}{\rmdefault}{\mddefault}{\updefault}0}}}
\put(10287,2925){\makebox(0,0)[lb]{{\SetFigFont{7}{8.4}{\rmdefault}{\mddefault}{\updefault}0}}}
\put(8487,75){\makebox(0,0)[lb]{{\SetFigFont{7}{8.4}{\rmdefault}{\mddefault}{\updefault}0}}}
\put(9087,75){\makebox(0,0)[lb]{{\SetFigFont{7}{8.4}{\rmdefault}{\mddefault}{\updefault}1}}}
\put(11487,75){\makebox(0,0)[lb]{{\SetFigFont{7}{8.4}{\rmdefault}{\mddefault}{\updefault}2}}}
\put(9687,75){\makebox(0,0)[lb]{{\SetFigFont{7}{8.4}{\rmdefault}{\mddefault}{\updefault}2}}}
\put(12087,75){\makebox(0,0)[lb]{{\SetFigFont{7}{8.4}{\rmdefault}{\mddefault}{\updefault}0}}}
\put(10287,75){\makebox(0,0)[lb]{{\SetFigFont{7}{8.4}{\rmdefault}{\mddefault}{\updefault}0}}}
\put(10887,75){\makebox(0,0)[lb]{{\SetFigFont{7}{8.4}{\rmdefault}{\mddefault}{\updefault}1}}}
\put(10887,2925){\makebox(0,0)[lb]{{\SetFigFont{7}{8.4}{\rmdefault}{\mddefault}{\updefault}1}}}
\path(12,675)(612,675)(612,225)
	(12,225)(12,675)
\put(12,375){\makebox(0,0)[lb]{{\SetFigFont{6}{7.2}{\rmdefault}{\mddefault}{\updefault} $W^{[22]}$}}}
\path(612,675)(1212,675)(1212,225)
	(612,225)(612,675)
\put(612,375){\makebox(0,0)[lb]{{\SetFigFont{6}{7.2}{\rmdefault}{\mddefault}{\updefault} $W^{[12]}$}}}
\path(1212,675)(1812,675)(1812,225)
	(1212,225)(1212,675)
\path(1812,675)(2412,675)(2412,225)
	(1812,225)(1812,675)
\put(1812,375){\makebox(0,0)[lb]{{\SetFigFont{6}{7.2}{\rmdefault}{\mddefault}{\updefault} $W^{[22]}$}}}
\path(2412,675)(3012,675)(3012,225)
	(2412,225)(2412,675)
\put(2412,375){\makebox(0,0)[lb]{{\SetFigFont{6}{7.2}{\rmdefault}{\mddefault}{\updefault} $W^{[12]}$}}}
\path(3012,675)(3612,675)(3612,225)
	(3012,225)(3012,675)
\path(4212,225)(3612,225)(3612,675)(4212,675)
\put(3687,375){\makebox(0,0)[lb]{{\SetFigFont{6}{7.2}{\rmdefault}{\mddefault}{\updefault}   $\ldots$}}}
\put(12,0){\makebox(0,0)[lb]{{\SetFigFont{7}{8.4}{\rmdefault}{\mddefault}{\updefault}0}}}
\put(612,0){\makebox(0,0)[lb]{{\SetFigFont{7}{8.4}{\rmdefault}{\mddefault}{\updefault}1}}}
\put(3012,0){\makebox(0,0)[lb]{{\SetFigFont{7}{8.4}{\rmdefault}{\mddefault}{\updefault}2}}}
\put(1212,0){\makebox(0,0)[lb]{{\SetFigFont{7}{8.4}{\rmdefault}{\mddefault}{\updefault}2}}}
\put(3612,0){\makebox(0,0)[lb]{{\SetFigFont{7}{8.4}{\rmdefault}{\mddefault}{\updefault}0}}}
\put(1812,0){\makebox(0,0)[lb]{{\SetFigFont{7}{8.4}{\rmdefault}{\mddefault}{\updefault}0}}}
\put(2412,0){\makebox(0,0)[lb]{{\SetFigFont{7}{8.4}{\rmdefault}{\mddefault}{\updefault}1}}}
\path(12,3600)(612,3600)(612,3150)
	(12,3150)(12,3600)
\put(12,3300){\makebox(0,0)[lb]{{\SetFigFont{6}{7.2}{\rmdefault}{\mddefault}{\updefault} $W^{[11]}$}}}
\path(612,3600)(1212,3600)(1212,3150)
	(612,3150)(612,3600)
\put(612,3300){\makebox(0,0)[lb]{{\SetFigFont{6}{7.2}{\rmdefault}{\mddefault}{\updefault} $W^{[21]}$}}}
\path(1212,3600)(1812,3600)(1812,3150)
	(1212,3150)(1212,3600)
\path(1812,3600)(2412,3600)(2412,3150)
	(1812,3150)(1812,3600)
\put(1812,3300){\makebox(0,0)[lb]{{\SetFigFont{6}{7.2}{\rmdefault}{\mddefault}{\updefault} $W^{[11]}$}}}
\path(2412,3600)(3012,3600)(3012,3150)
	(2412,3150)(2412,3600)
\put(2412,3300){\makebox(0,0)[lb]{{\SetFigFont{6}{7.2}{\rmdefault}{\mddefault}{\updefault} $W^{[21]}$}}}
\path(3012,3600)(3612,3600)(3612,3150)
	(3012,3150)(3012,3600)
\path(4212,3150)(3612,3150)(3612,3600)(4212,3600)
\path(912,2250)(87,2925)
\path(198.872,2872.230)(87.000,2925.000)(160.878,2825.793)
\path(912,2250)(687,2850)
\path(757.225,2748.174)(687.000,2850.000)(701.045,2727.107)
\path(912,2250)(3537,2925)
\path(3428.252,2866.060)(3537.000,2925.000)(3413.310,2924.170)
\dottedline{45}(987,2625)(1737,2625)
\put(3762,3300){\makebox(0,0)[lb]{{\SetFigFont{6}{7.2}{\rmdefault}{\mddefault}{\updefault}  $\ldots$}}}
\put(612,2925){\makebox(0,0)[lb]{{\SetFigFont{7}{8.4}{\rmdefault}{\mddefault}{\updefault}1}}}
\put(3012,2925){\makebox(0,0)[lb]{{\SetFigFont{7}{8.4}{\rmdefault}{\mddefault}{\updefault}2}}}
\put(1212,2925){\makebox(0,0)[lb]{{\SetFigFont{7}{8.4}{\rmdefault}{\mddefault}{\updefault}2}}}
\put(3612,2925){\makebox(0,0)[lb]{{\SetFigFont{7}{8.4}{\rmdefault}{\mddefault}{\updefault}0}}}
\put(2412,2925){\makebox(0,0)[lb]{{\SetFigFont{7}{8.4}{\rmdefault}{\mddefault}{\updefault}1}}}
\put(12,2925){\makebox(0,0)[lb]{{\SetFigFont{7}{8.4}{\rmdefault}{\mddefault}{\updefault}0}}}
\put(1812,2925){\makebox(0,0)[lb]{{\SetFigFont{7}{8.4}{\rmdefault}{\mddefault}{\updefault}0}}}
\put(537,2025){\makebox(0,0)[lb]{{\SetFigFont{7}{8.4}{\rmdefault}{\mddefault}{\updefault}time in modulo 3}}}
\path(4812,3450)(7812,450)
\path(7705.934,513.640)(7812.000,450.000)(7748.360,556.066)
\path(4812,3450)(7812,3450)
\path(7692.000,3420.000)(7812.000,3450.000)(7692.000,3480.000)
\path(4812,450)(7812,450)
\path(7692.000,420.000)(7812.000,450.000)(7692.000,480.000)
\path(6282,4290)(6057,3840)
\path(6083.833,3960.748)(6057.000,3840.000)(6137.498,3933.915)
\path(7748.360,3343.934)(7812.000,3450.000)(7705.934,3386.360)
\path(7812,3450)(4812,450)
\put(6162,600){\makebox(0,0)[lb]{{\SetFigFont{7}{8.4}{\rmdefault}{\mddefault}{\updefault}2}}}
\put(5637,1500){\makebox(0,0)[lb]{{\SetFigFont{7}{8.4}{\rmdefault}{\mddefault}{\updefault}1}}}
\put(5787,2700){\makebox(0,0)[lb]{{\SetFigFont{7}{8.4}{\rmdefault}{\mddefault}{\updefault}0}}}
\put(6012,3600){\makebox(0,0)[lb]{{\SetFigFont{7}{8.4}{\rmdefault}{\mddefault}{\updefault}0}}}
\put(5637,4425){\makebox(0,0)[lb]{{\SetFigFont{7}{8.4}{\rmdefault}{\mddefault}{\updefault}Propogation delay (modulo 3)}}}
\put(4362,375){\makebox(0,0)[lb]{{\SetFigFont{7}{8.4}{\rmdefault}{\mddefault}{\updefault}Tx 2}}}
\put(7887,3375){\makebox(0,0)[lb]{{\SetFigFont{7}{8.4}{\rmdefault}{\mddefault}{\updefault}Rx 1}}}
\put(7887,375){\makebox(0,0)[lb]{{\SetFigFont{7}{8.4}{\rmdefault}{\mddefault}{\updefault}Rx 2}}}
\put(4362,3375){\makebox(0,0)[lb]{{\SetFigFont{7}{8.4}{\rmdefault}{\mddefault}{\updefault}Tx 1}}}
\end{picture}
}

%% file: parallel_relay.eepic
\setlength{\unitlength}{0.00045833in}
\begingroup\makeatletter\ifx\SetFigFont\undefined%
\gdef\SetFigFont#1#2#3#4#5{%
  \reset@font\fontsize{#1}{#2pt}%
  \fontfamily{#3}\fontseries{#4}\fontshape{#5}%
  \selectfont}%
\fi\endgroup%
{\renewcommand{\dashlinestretch}{30}
\begin{picture}(8337,5636)(0,-10)
\put(4387,3288){\ellipse{524}{4650}}
\path(600,4338)(975,4338)
\path(855.000,4308.000)(975.000,4338.000)(855.000,4368.000)
\path(1575,4338)(4350,5238)
\path(1650,2913)(4350,5238)
\path(4350,4863)(4425,4863)(4425,4788)
	(4350,4788)(4350,4863)
\path(4350,5313)(4425,5313)(4425,5238)
	(4350,5238)(4350,5313)
\path(4350,1413)(4425,1413)(4425,1338)
	(4350,1338)(4350,1413)
\path(1575,4338)(4350,4788)
\path(1650,2913)(4350,1338)
\path(1650,2913)(4350,4788)
\path(4350,1938)(4425,1938)(4425,1863)
	(4350,1863)(4350,1938)
\path(1650,2913)(4350,1863)
\path(1575,4338)(4350,1863)
\path(1575,4338)(4350,1338)
\path(1650,2913)(4425,2913)
\path(4425,2913)(7200,2913)
\path(1575,4338)(4350,4338)
\path(4425,4338)(7200,4338)
\path(4425,5238)(7200,4338)
\path(4425,4788)(7200,4338)
\path(4425,1863)(7200,4338)
\path(4425,1338)(7200,4338)
\path(4425,4788)(7200,2913)
\path(4425,5238)(7200,2913)
\path(4425,1863)(7200,2913)
\path(4425,1338)(7200,2913)
\path(675,2913)(1050,2913)
\path(930.000,2883.000)(1050.000,2913.000)(930.000,2943.000)
\path(5400,288)(4500,1038)
\path(4611.392,984.225)(4500.000,1038.000)(4572.981,938.131)
\path(7950,4338)(8325,4338)
\path(8205.000,4308.000)(8325.000,4338.000)(8205.000,4368.000)
\path(7875,2913)(8250,2913)
\path(8130.000,2883.000)(8250.000,2913.000)(8130.000,2943.000)
\thicklines
\dottedline{165}(4425,4713)(4425,2013)
\put(975,4263){\makebox(0,0)[lb]{{\SetFigFont{7}{8.4}{\rmdefault}{\mddefault}{\updefault}$X^{[1]}$}}}
\put(1125,2838){\makebox(0,0)[lb]{{\SetFigFont{7}{8.4}{\rmdefault}{\mddefault}{\updefault}$X^{[2]}$}}}
\put(75,2838){\makebox(0,0)[lb]{{\SetFigFont{7}{8.4}{\rmdefault}{\mddefault}{\updefault}$W^{[22]}$}}}
\put(7350,2838){\makebox(0,0)[lb]{{\SetFigFont{7}{8.4}{\rmdefault}{\mddefault}{\updefault}$Y^{[2]}$}}}
\put(7350,4263){\makebox(0,0)[lb]{{\SetFigFont{7}{8.4}{\rmdefault}{\mddefault}{\updefault}$Y^{[1]}$}}}
\put(4725,63){\makebox(0,0)[lb]{{\SetFigFont{8}{9.6}{\familydefault}{\mddefault}{\updefault}$K$ distributed relays}}}
\put(8250,2838){\makebox(0,0)[lb]{{\SetFigFont{7}{8.4}{\rmdefault}{\mddefault}{\updefault}$\widehat{W}^{[22]}$}}}
\put(8325,4263){\makebox(0,0)[lb]{{\SetFigFont{7}{8.4}{\rmdefault}{\mddefault}{\updefault}$\widehat{W}^{[11]}$}}}
\put(0,4263){\makebox(0,0)[lb]{{\SetFigFont{7}{8.4}{\rmdefault}{\mddefault}{\updefault}$W^{[11]}$}}}
\end{picture}
}